\documentclass[a4paper,12pt]{article}
\usepackage{latexsym,amsmath,amssymb, amsfonts,mathrsfs}
\usepackage{cite}
\usepackage[arrow,matrix,curve]{xy}
\usepackage{amsmath}

\usepackage{mathrsfs}
\DeclareMathAlphabet{\mathpzc}{OT1}{pzc}{m}{it}

\usepackage{pdfsync}

\usepackage{hyperref}
\usepackage{graphicx}
\usepackage{showidx}
\usepackage{cite}

\usepackage{amssymb}
\usepackage{amscd}
\usepackage{amsfonts}
\usepackage{amsthm}

\numberwithin{equation}{section}

\theoremstyle{definition}
\newtheorem{Theorem}{Theorem}[]

\newtheorem{Proposition}[Theorem]{Proposition}
\newtheorem{Corollary}[Theorem]{Corollary}

\theoremstyle{definition}

\newtheorem{Example}[Theorem]{Examples}

\newtheorem{Note}[Theorem]{Note}

\DeclareMathAlphabet{\mathpzc}{OT1}{pzc}{m}{it}

\global\arraycolsep=2pt 
\def\s[#1,#2]{[#1\stackrel{{\displaystyle\star}}{,}#2]}

 1

\newcommand{\eq}{\begin{equation}}
\newcommand{\eqa}{\begin{eqnarray}}
\newcommand{\en}{\end{equation}}
\newcommand{\ena}{\end{eqnarray}}
\newcommand{\enn}{\nonumber \end{equation}}

\def\sk{\vskip .4cm}
\def\noi{\noindent}

\def\epsi{{\varepsilon}}
\def\st {\star}

\def\D/h{\widehat{\fmslash D}}

\def\al{\alpha}
\def\la{\lambda}
\def\be{\beta}

\def\5bar{{\overline 5}}

\def\FF{\mathcal F}

\def\s'O{\stackrel{_{{\displaystyle\st \footnotesize '}}}{_{^{^{\displaystyle\otimes}}}}}

\def\D{\Delta}
\def\1s{{1_\st }}
\def\3s{{3_\st }}
\def\2s{{2_\st }}
\def\ef1{{1_\FF}}
\def\ef2{{3_\FF}}
\def\ef3{{2_\FF}}

\def\hbar{\lambda}

\def\cc{\mathbb{C}}

\newcommand{\arxiv}[1]{{\tt
\href{http://www.arXiv.org/abs/#1}{arXiv:#1}}}

\def\Tbarp{{\bar T}{^{\mbox{${\;\!}^+$}}\!}}
\def\Tbarm{{\bar T}{^{\mbox{${\;\!}^-$}}\!}}

\newcommand{\del}{\partial}
\newcommand{\pa}{\partial}
\newcommand{\eqn}[1]{(\ref{#1})}

\newcommand{\nn}{\nonumber}

\newcommand{\N}{{\cal N}}
\newcommand{\cN}{{\cal N}}
\newcommand{\M}{{\cal M}}
\newcommand{\cM}{{\cal M}}

\newcommand{\A}{{\cal A}}

\newcommand{\matc}{\begin{array}{c}}
\newcommand{\matcc}{\begin{array}{cc}}
\newcommand{\matccc}{\begin{array}{ccc}}
\newcommand{\matcccc}{\begin{array}{cccc}}
\newcommand{\emat}{\end{array}}

\newcommand{\uuu}{u}

\def\bar{\overline}

\newcommand{\HF}{{{\,\,F^{}_{}}^{\!\:\!\!\!\!\!\!\!\!\!{{{\ast}}}~~}}}
\newcommand{\HG}{{{{\,\,G^{}_{}}}^{_{\:}\!\!\!\!\!\!\!\!\!\!\!{{{\ast}}}~~}}}
\newcommand{\HT}{{{_{\,}\,\,T^{}_{}}^{\!\!\;\!\!\!\!\!\!\!\!\!{{{\ast}}}~~}}}
\newcommand{\HbT}{{~\,{\bar T}^{\!\!\!\!\!\!\!\!\!\ast\,~~}}}
\begin{document}


\begin{titlepage}

{\vspace{-2.8em}}

\hfill {CERN-PH-TH/2013-004}
\sk\sk

\begin{center}
\sk
\sk
{\bf \large Constitutive relations and Schr\"odinger's formulation of \\[.4em]nonlinear
   electromagnetic theories}

\sk\sk
{\bf Paolo Aschieri$^{1,2}$ and Sergio Ferrara$^{3,4}$}
\sk

{\it $^1$Dipartimento di Scienze e Tecnologie
 Avanzate, Universit\`{a} del
 Piemonte Orientale,}\\  {\it $^2$INFN, Sezione di Torino, gruppo collegato di Alessandria }\\
{\it Viale T. Michel 11, 15121 Alessandria, Italy}\\
{\small{\texttt{aschieri@to.infn.it}}}\\[.5em]

        {\it $^3$Physics Department,Theory Unit, CERN, 
        CH 1211, Geneva 23, Switzerland}\\
        {\it $^4$INFN - Laboratori Nazionali di Frascati, 
        Via Enrico Fermi 40,I-00044 Frascati, Italy}\\
{\small\texttt{sergio.ferrara@cern.ch}}\\[.5em]
\sk\sk\sk

\sk\sk
\begin{abstract} 
We present a systematic study of nonlinear and higher derivatives
extensions of electromagnetism.
We clarify when action functionals $S[F]$ can be explicitly obtained from
arbitrary  (not necessarily self-dual) nonlinear equations of motion.
We show that the  ``Deformed twisted self-duality condition''
proposal originated in the context of supergravity counterterms is
actually the general framework needed to discuss self-dual theories
starting from a variational principle. 

We generalize to nonlinear and higher derivatives
theories  Schr\"odinger formulation of Born-Infeld theory, and for the
latter, and more in general for nonlinear theories,  we derive a closed form
expression of the corresponding deformed twisted self-duality
conditions.  
This implies that the hypergeometric expression  entering these duality conditions and
leading to  Born-Infeld theory satisfies a hidden quartic equation.
\end{abstract}
\sk
\setcounter{page}{0}
\end{center}
\end{titlepage}

\renewcommand{\thepage}{\arabic{page}}

\def\thefootnote{\arabic{footnote}} \setcounter{footnote}{0}


\section{Introduction}
Duality is a leading paradigm of theoretical physics. Electric-magnetic duality is
one of the oldest and most studied examples. Maxwell
theory is self-dual, i.e., admits duality symmetry under rotation of
the electric field into the magnetic one. 
Schr\"odinger \cite{Schrodinger} was the first to show that the
nonlinear theory of electromagnetism of Born and Infeld, quite
remarkably has the same $U(1)$ duality symmetry property.
The study of electric-magnetic duality symmetry  has found
further motivations since its appearance in extended supergravity
theories \cite{FSZ77,csf,crju}.  In \cite{csf} the first example of
a noncompact duality rotation group was considered, it arises in $N=4$
supergravity and is due to scalar fields transforming under duality
rotations.
These results triggered
further investigations in the general structure of  self-dual
theories. In particular the symplectic formalism for nonlinear
electromagnetism coupled to scalar and
fermion fields was initiated in \cite{GZ}, there  the
duality  groups were shown to be subgroups of noncompact symplectic
groups (the compact case being recovered in the absence of scalar
fields).  A nonlinear example is  Born-Infeld electrodynamics
coupled to axion and dilaton fields \cite{Gibbons:1995ap}.
Another relevant aspect \cite{BG} is that the spontaneous breaking of $N=2$ rigid
supersymmetry to $N=1$ can lead to a Goldstone vector multiplet whose
action is the supersymmetric and self-dual Born-Infeld action 
\cite{DP, CF}.
Higher supersymmetric Born-Infeld type actions are also self-dual and related to spontaneous
supersymmetry breakings in field theory \cite{KET, KT, KT2, BIK} and in string
theory \cite{KET2, RT}.

\sk

Duality symmetry is a powerful tool to investigate the
structure of possible counterterms in extended supergravity. 
After the explicit computations that showed the 3-loop UV finiteness of 
$N=8$ supergravity \cite{Bern}, an explanation based on  
$E_{7(7)}$ duality symmetry was provided 
\cite{Brodel:2009hu, Elvang:2010kc, Beisert:2010jx, Bossard:2010bd}.
Furthermore
duality symmetry arguments have also been used to suggest
all loop finiteness of $N=8$ supergravity \cite{Kallosh:2011dp}.
Related to these developments, 
in \cite{BN} a proposal on how to implement 
duality rotation invariant counterterms in a corrected action $S[F]$ leading to a
self-dual theory was put forward under the name of ``deformed 
 twisted self-duality conditions'' (see eq. (\ref{Iconst})).
Examples included
counterterms dependent on derivatives
of the field strength. The proposal (renamed ``nonlinear 
twisted self-duality conditions'') was further elaborated in
\cite{CKR} and \cite{CKO}; see also \cite{BCFKR}, and 
\cite{Kuzenko, Kuzenko:2013gr}, for the supersymmetric extensions and examples.
The proposal is equivalent to a formulation of self-dual theories using auxiliary fields studied in \cite{IZ2001} and \cite{Ivanov:2003uj} in case of nonlinear electromagnetism without higher derivatives of the
field strength. 
This coincidence has been brought to light in a very recent paper \cite{IZ}.

The supergravity motivated studies have provided new examples of
self-dual theories and have touched upon basic issues
 like consistency and equivalence of different formulations of self-duality
conditions, reconstruction of the action from these conditions and of
duality invariant expressions. This paper is a  systematic study of
these issues. 
\sk
A nonlinear and higher derivative
electromagnetic theory is determined by defining, eventually
implicitly, the relation between the electric field strength $F$
(given by the electric field $E$ and the magnetic induction
$B$ ) and the  magnetic field strength $G$ (given by the magnetic field $H$ 
and the electric displacement $D$). 
We call {\sl constitutive relations} the relations defining $G$ in
terms of $F$ or vice versa. 

We begin Section  \ref{dualityrot}  by proving that (locally) the equations
of motions of an arbitrary, not necessarily self-dual, nonlinear
electromagnetic theory  satisfying an integrability condition can 
always be obtained from a variational principle via an action  $S[F]$
that is explicitly computed (reconstructed) from the constitutive
relations.
This reconstruction procedure works also for theories
with higher derivatives if we further assume that they 
can be obtained from an action principle.

We then study  the general theory
of $U(1)$ duality rotations.
Self-duality of the equations of motion constrains the constitutive
relations. The  deformed twisted self-duality
conditions are just constitutive relations obtained from a variational
procedure. In these  deformed twisted self-duality
conditions the dependence of $G$ from $F$ is given implicitly, but the
constraint that leads to self-dual theories is easily implemented.
This is due to the use of the complex and chiral variables $T^+$,
$T^-$,  $\overline{T^+}$,  $\overline{T^-}$ that are the chiral
projections of the variables  $T=F-iG$  and $\overline
T=F+iG$ introduced by Schr\"odinger \cite{Schrodinger, GZS}.
The fields  $T^+$, $T^-$,  $\overline{T^+}$,  $\overline{T^-}$ have definite
electric-magnetic duality charge and chirality: $(T^+,+1,+1),
~(T^-,+1,-1), ~(\overline{T^+},-1,-1),  ~(\overline{T^-},-1,+1)$. 

The action $S[F]$ can
always be reconstructed from the action $\cal I[T^-,\overline{T^-}]$ that determines the 
deformed twisted self-duality conditions, and vice versa. Indeed,
as also shown in \cite{Ivanov:2003uj}, the
two actions are related by a Legendre transformation. This shows that
the 
deformed twisted self-duality conditions are the 
general framework needed to discuss self-dual theories obtained from a
variational principle.

Section \ref{constitutiverelations} is devoted to a detailed study of
the constitutive relations of the kind 
${\HG}{}_{\mu\nu}={\cN_2}_{\,} F_{\mu\nu} +{\cN_1}_{\,}
\HF{}_{\!\mu\nu}$ where ${\cN_1}$ and ${\cN_2}$ are real
(pseudo)scalar functions of $F$, $G$ and their derivatives.
These are not the most general constitutive relations because 
${\cN_1}$ and ${\cN_2}$ are not differential operators and do not
act on the $\mu\nu$-indices of  $F_{\mu\nu}$
and its Hodge dual $\HF{}_{\!\mu\nu}$. However they describe a wide class of
nonlinear theories. For example theories without higher
derivatives are determined by this kind of
relations.\footnote{Indeed in this case  the elementary
  antisymmetric 2-tensors in the theory are  only $F_{\mu\nu}$ and its Hodge dual $\HF_{\mu\nu}$, hence any
antisymmetric 2-tensor will be a linear combination (with
coefficients
dependent on the field strength)
of  $F_{\mu\nu}$ and $\HF_{\mu\nu}$.}
Equivalent but more duality symmetric formulations of these constitutive relations are
then investigated. In particular we formulate consistent constitutive
relations in terms of the complex variables $T=F-iG$  and $\overline
T=F+iG$, thus generalizing  Schr\"odinger study of  Born-Infeld theory \cite{Schrodinger, GZS}.

In Section \ref{SCH} the constitutive relations of Section \ref{constitutiverelations} 
are constrained to define self-dual theories. 
These self-dual constitutive relations turn out to be very simple. They
are determined for example by expressing the  ratio 
$\frac{T_{\mu\nu}\overline T^{\mu\nu}}{|T_{\mu\nu\,}\HT^{\mu\nu}|}$
in terms of $T,\, \overline T$ and their derivatives.
In particular we see that self-duality constraints the phases of 
$T_{\mu\nu}\HT^{\mu\nu}$ and $T_{\mu\nu}T^{\mu\nu}$ to differ by a $-\pi/2$ angle and
the square of their moduli to differ by $|T_{\mu\nu}\overline T^{\mu\nu}|^2$.

Section \ref{nhd} considers self-dual theories that do not involve higher
derivatives of the field strength. In this case the natural
independent variable is $|T_{\mu\nu}T^{\mu\nu}|$. We present a
closed form expression of the {deformed twisted self-duality
  conditions} that determine  Born-Infeld theory.
Comparison of this expression with the one in terms of a
hypergeometric function ${\mathfrak F}$  previously considered in
\cite{CKR} leads to a hidden quartic equation for ${\mathfrak F}$.
This quartic equation is not just a feature of Born-Infeld theory. It also enters  the explicit relation we obtain
between deformed twisted self-duality conditions of any nonlinear
theory and the corresponding  constitutive relations in
the Schr\"odinger's variables $T$,$\overline T$.

In the appendices we provide examples of self-dual theories with higher
derivatives, a basic result on the energy momentum tensor
 of nonlinear theories and details on a technical calculation.

\section{U(1) duality rotations in nonlinear and higher derivatives
  electromagnetism \label{dualityrot}}
\subsection{Action functionals from equations of motion\label{AFEOM}}

Nonlinear and higher derivatives electromagnetism is described by the equations of motion
\eqa
&&{\pa}_{\mu}
{\widetilde F}^{\mu\nu}  =0~,\label{max22}\\
&&{\pa}_{\mu}
\widetilde{G}^{\mu\nu}=0~, \label{max11}\\
&&
\widetilde G^{\mu\nu}=h^{\mu\nu}[F,\la]
\label{maxwww}~.
\ena 
The first two simply state that the 2-forms $F$ and $G$ are closed, ${{d}} F={{d}} G=0$, indeed
$\widetilde
F^{\mu\nu}\equiv\frac{1}{2}\epsi^{\mu\nu\rho\sigma}F_{\rho\sigma}$,
$\widetilde
G^{\mu\nu}\equiv\frac{1}{2}\epsi^{\mu\nu\rho\sigma}G_{\rho\sigma}$
(with $\epsi^{0123}=1$). The last set 
$\widetilde G^{\mu\nu}=h^{\mu\nu}[F,\la]$, where $\la$ is the
dimensionful parameter typically present in a nonlinear theory\footnote{Nonlinear and higher derivatives theories of electromagnetism  admit
one (or more) dimensionful coupling constant(s) $\la$. Since the 
expansions for weak and slowly varying fields are expansions in  
adimensional variables (like for example $\la F F$ and $\la F
\HF$, or, schematically
and using a
different coupling constant, $\la\partial F\partial F$) 
we will equivalently say that these expansions are in power 
series of the coupling constant(s) $\la$.
}, 
are the constitutive relations.  They specify the dynamics and
determine the magnetic field strength $G$ as a functional in term of the electric field strength $F$, and, vice versa, determine $F$ in term
of $G$, indeed $F$ and $G$ should be treated on equal footing in (\ref{max22})-(\ref{maxwww}).
The square bracket notation  $h^{\mu\nu}[F,\la]$ stems from
the possible dependence of $h^{\mu\nu} $ from derivatives of $F$.

Since in general we consider  curved background metrics
$g_{\mu\nu}$, it is convenient to introduce the $\ast$-Hodge operator;
on an arbitrary antisymmetric tensor $F_{\mu\nu}$ it is defined by 
\eq
\HF{}_{\mu\nu}=
\frac{1}{2\sqrt{g}}
g_{\mu\al}g_{\nu\beta}\,\epsi^{\al\be\rho\sigma}F_{\rho\sigma}
=\frac{1}{\sqrt{g}}\widetilde F_{\mu\nu}~,
\en
where $g=-\det(g_{\mu\nu})$, and it squares to
minus the identity.
The constitutive relations (\ref{maxwww}) implicitly include also
a dependence on the background metric $g_{\mu\nu}$ and for example in
case of usual electromagnetism they read $G_{\mu\nu}=\HF_{\mu\nu}=\frac{1}{\sqrt{g}}\widetilde
F_{\mu\nu}$, while for 
Born-Infeld theory, 
\eq
{S}_{BI}= \frac{1}{\la}\int\!d^4x\,\sqrt{g}\Big( 1-\sqrt{1+\frac{1}{2}\la
  F^2-\frac{1}{16}\la^2(F\HF)^2}\;\Big)~,\label{BILag}
\en
where $F^2=FF=F_{\mu\nu}F^{\mu\nu}$ and 
$F\HF=F_{\mu\nu}\HF^{\mu\nu}$,
they read
\eq\label{BIcr}
{G}_{\mu\nu}=
\frac{\HF{}_{\!\mu\nu}+{1\over 4} \la (F{\HF})\,F_{\mu\nu}}{
\sqrt{1+{1\over 2}\la F^2-\frac{1}{16}\la^2(F{\HF})^2}}~.
\en
The constitutive relations (\ref{maxwww}) define a nonlinear and higher
derivative extension of electromagnetism because we require that setting $\la=0$ in
(\ref{maxwww})
we recover usual electromagnetism: $G_{\mu\nu}=\HF{}_{\!\mu\nu}$.
\sk
We now show that in the general nonlinear case (where the constitutive relations do not
involve derivatives of $F$) the equations of motion  (\ref{max22})-(\ref{maxwww})
can always be obtained from a variational principle provided they
satisfy the integrability conditions 
\eq\label{intcond}
\frac{\partial
{h}^{\mu\nu}}{\partial F_{\rho\sigma}}=\frac{\partial
{h}^{\rho\sigma}}{\partial F_{\mu\nu}}~.
\en
These conditions are necessary in order to obtain (\ref{maxwww}) from
an action $S[F]=\int \!d^4x \/{\cal L}(F)$. Indeed if\/\footnote{The factor 2 is due to the 
convention $\frac{\partial{F_{\rho\sigma}}}{\partial
  F_{\mu\nu}}=\delta^\mu_\rho\delta^\nu_\sigma\,$ adopted in \cite{GZ}
and in the review \cite{AFZ}. It will be used
throughout the paper.\label{funo}} $h^{\mu\nu}=2\frac{\partial
{\cal L}}{\partial F_{\mu\nu}}$ then (\ref{intcond}) trivially holds.

In order to show that (\ref{intcond}) is also sufficient we recall that the field
strength $F_{\mu\nu}(x)$ locally is a map from spacetime to
$\mathbb{R}^6$ (with coordinates $F_{\mu\nu}$, {\small{$\mu<\nu$}}). We assume 
$h^{\mu\nu}(F,\la)$ to be  well defined functions
on $\mathbb{R}^6$ or more in general on an open submanifold $M\subset
\mathbb{R}^6$ that includes the origin ($F_{\mu\nu}=0$) and that is a star shaped
region w.r.t. the origin (e.g. a 6-dimensional ball or cube
centered in the origin).

Then condition (\ref{intcond}) states that the 1-form
$\mathpzc{h}=h^{\mu\nu}dF_{\mu\nu}$ is closed, and hence, by Poincar\'e lemma,
exact on $M$; we write $\mathpzc{h}=d{\cal L}$. We have ${\cal L}(F)-{\cal L}(0)=\int_\gamma {}_{\!}\mathpzc{h}\,$ for
any curve $\gamma(c)$ of coordinates $\gamma_{\mu\nu}(c)$ such that
$\gamma_{\mu\nu}(0)=0$ and $\gamma_{\mu\nu}(1)=F_{\mu\nu}$. In
particular, choosing the straight line from the origin to the point of
coordinates $F_{\mu\nu}$,
and setting $S=\int d^4x \,{\cal L}(F)$, we immediately obtain

\begin{Theorem}\label{actionfromeom}
If the constitutive relations (\ref{maxwww}) do not depend on
derivatives of $F$ (i.e. if $h^{\mu\nu}[F,\la]=h^{\mu\nu}(
F,\la)\,$) and the functions $h^{\mu\nu}(F,\la)$ are defined in a
star shaped region $M$ (of coordinates $F_{\mu\nu}$) and satisfy the integrability conditions (\ref{intcond}),
then the constitutive relations (\ref{maxwww})  are equivalent to the equations\/$^{\ref{funo}}$
 \eq
{\widetilde G}^{\;\mu\nu}= 2\frac{\delta S[F]}{\delta F_{\mu\nu}} ~\label{Sconst0}
\en
where the action functional $S[F]$ is given by
\eq
S[F]=\frac{1}{2}\int \!d^4x_{\,} F_{\mu\nu}\!\!\int_0^1\! dc \, h^{\mu\nu}(cF,\la)~\label{recide}.
\en
\end{Theorem}

\begin{Corollary} On spacetimes where closed two forms are exact
($dF=0\Rightarrow F=dA$),  the
 equations of motion (\ref{max22})-(\ref{maxwww}) of
nonlinear electromagnetism satisfying the conditions of Theorem \ref{actionfromeom}
are equivalent to the equations of motion
\eq
\frac{\delta S}{\delta A_{\mu}}=0 ~\label{Aeom}
\en
where $S=\frac{1}{2}\int \!d^4x_{} \int_0^1\! dc \, F_{\,}
h(cF,\la)$. 
\end{Corollary}
\begin{proof} Equation (\ref{max22}) is  the
Bianchi identity for $F=dA$, (\ref{maxwww}) holds because of 
Theorem \ref{actionfromeom}, and  (\ref{max11}) is equivalent to the equations of
motion (\ref{Aeom}).
\end{proof}
We have seen that under the integrability conditions (\ref{intcond})
locally  the equations of motion of nonlinear electromagnetism 
(\ref{max22})-(\ref{maxwww})  can be obtained  from the action 
\eq
S=\frac{1}{2}\int \!d^4x_{} \int_0^1\! dc \,c F_{\,} \widetilde G_c~,
\en 
where $\widetilde G_c=\frac{1}{c} h(cF,\la)$.

It is interesting to generalize  these results to the case of nonlinear
and higher derivatives electromagnetism. We here present a first step in this
direction

\begin{Proposition}\label{fstep}
If the equations of motion (\ref{max22})-(\ref{maxwww}) of a nonlinear and higher derivatives
electromagnetic theory are obtained from an action functional $S[F]$
then we have 
\eq
S[F]=\frac{1}{2}\int \!d^4x_{} \int_0^1\! dc \, F_{\,} h[cF,\la]~,
\en
that we simply rewrite  $S=\frac{1}{2}\int \!d^4x_{} \int_0^1\! dc \,c
F_{\,} \widetilde G_{c}$.
\end{Proposition}
\begin{proof}
Consider the one parameter family of actions
$S_c[F]=\frac{1}{c^2}S[cF]$. 
Deriving with respect to $c$ we obtain 
\eq
-c\frac{\partial S_c}{\partial c}=2S_c-\int \!d^4x ~F\frac{\delta
S_c[F]}{\delta F}~,\label{Ttrace}
\en 
i.e. $-c\frac{\partial S_c}{\partial c}=2S_c-\frac{1}{2}\int \!d^4x
~F\widetilde G_c$. It is easy to see that 
$S_c=\frac{1}{2c^2}\int \!d^4x_{} \int_0^c\! dc' \,c'
F_{\,} \widetilde G_{c'}$ is the primitive with the correct behaviour
under rescaling of $c$ and $F$. We conclude that 
$\frac{1}{c^2}S[cF]= \frac{1}{2c^2}\int \!d^4x_{} \int_0^c\! dc' \,c'
F_{\,} \widetilde G_{c'}$, and setting $c=1$ we obtain the thesis.
\end{proof}
\sk
We now consider the following expansion of an action $S[F]$ 
even under $F\to -F$,
\eq
S[F]=S^{\{0\}}[F]+S^{\{2\}}[F]+S^{\{4\}}[F]+\ldots\label{FexpansionS}
\en
where $S^{\{2n\}}$ is the term homogeneous in
$2n$ field strengths or their derivatives. Similarly we consider 
$S_c[F]=\frac{1}{c^2}S[cF]$ and expand 
$\widetilde G_c=2\frac{\delta S_c}{\delta F}$ as
\eqa
\widetilde G_c&=&\widetilde G_c^{\{1\}}+\widetilde
G_c^{\{3\}}+\widetilde G_c^{\{5\}}+\ldots~\nonumber\\
&=&\widetilde G^{\{1\}}+c^2\widetilde
G^{\{3\}}+c^4\widetilde G^{\{5\}}+\ldots~\label{FexpansionG}
\ena
where  $G_c^{\{2n-1\}}$ is the term homogeneous in
$2n-1$ field strengths or their derivatives, and in the second
equality we observed that it is also the
term proportional to
$c^{2n-2}$ so that $G_c^{\{2n-1\}}=c^{2n-2}G_{c=1}^{\{2n-1\}}=c^{2n-2}G^{\{2n-1\}}$.
Proposition \ref{fstep} then implies
\eq
S^{\{2n\}}=\frac{1}{4n}\int d^4x\, F\widetilde G^{\{2n-1\}}~.\label{2.20}
\en
This expression relates the term in the action proportional to the
$2n^{\rm th}$ power of $F$ or its derivatives, to the term in $\widetilde G$ proportional
to the $(2n-1)^{\rm th}$ power of $F$ or its derivatives.

\begin{Note}\label{note3'} Expression  $S=\frac{1}{2}\int \!d^4x_{} \int_0^1\! dc \,c
F_{\,} \widetilde G_{c}$, in the equivalent form 
\eq
S=\frac{1}{4}\int \!d^4x \int_0^1 d\kappa \,F\widetilde G_{\kappa}~\label{recid}
\en
\noi
(where $\kappa=c^2$) has been considered for self-dual theories in \cite{CKO} and called reconstruction identity.
It has been used, together with an expression equivalent to (\ref{2.20}), to reconstruct the action $S$ from  equations of motion with duality
rotation symmetry in examples with higher derivatives of $F$. 
\end{Note}
\begin{Note}\label{note3''} In Appendix \ref{TTT} we show that for nonlinear
theories without higher derivatives, the l.h.s.~and r.h.s of (\ref{Ttrace}) are
half the spacetime integral of the trace of the energy momentum
tensor. 
\end{Note}

\subsection{$U(1)$ duality rotations}
Nonlinear and higher derivatives electromagnetism admits $U(1)$ duality rotation symmetry if 
given  a field configuration $F,G$ that satisfies
(\ref{max22})-(\ref{maxwww}) then the rotated configuration
\eq\label{rotFG}
\left(\begin{array}{c}
F' \\
G'
\end{array}\right)=
\left(\begin{array}{cc}
\cos\al & -\sin {\al}\\
\sin\al & \cos\al
\end{array}\right)
\left(\begin{array}{c}
F  \\
G
\end{array}\right)~,
\en
that is trivially a solution of 
${\pa}_{\mu}
{\widetilde F}^{\mu\nu}  =0\,,\;
{\pa}_{\mu}
\widetilde{G}^{\mu\nu}=0\,,
$
satisfies also 
$\widetilde G'_{\mu\nu}=h_{\mu\nu}[F',\la]$, so that $F',G'$ is again a solution of
the  equations of motion.
If we consider an infinitesimal duality rotation, $F\to F+\Delta F$,
$G\to G+\Delta G$ then condition 
$\widetilde G'_{\mu\nu}=h_{\mu\nu}[F',\la]$ reads 
$\Delta\widetilde G_{\mu\nu}=
\int d^4x\, \frac{\delta
  h_{\mu\nu}}{\delta F_{\rho\sigma}}\,\Delta F^{\rho\sigma}$,
i.e., $\widetilde F_{\mu\nu}=-\int d^4x\, \frac{\delta
  h_{\mu\nu}}{\delta F_{\rho\sigma}}\,G^{\rho\sigma}$, that we simply rewrite
\eq\label{basicDR}
\widetilde F_{\mu\nu}=-\int d^4x\, \frac{\delta \widetilde G_{\mu\nu}}{\delta F_{\rho\sigma}}\,G^{\rho\sigma}~.
\en
It is straightforward to check that electromagnetism and Born-Infeld
theory satisfy (\ref{basicDR}).
\sk
If the theory is obtained from an action
functional $S[F]$ (in the field strength $F$ and its derivatives) then
(\ref{maxwww}) is given by
 \eq
\widetilde G^{\mu\nu}= 2\frac{\delta S[F]}{\delta F_{\mu\nu}} ~.\label{Sconst}
\en
In particular it  follows that  
\eq
\frac{\delta{\widetilde G}^{\mu\nu}}{\delta F_{\rho\sigma}}=\frac{\delta
{\widetilde G}^{\rho\sigma}}{\delta F_{\mu\nu}}~,
\en
hence the duality symmetry condition (or self-duality condition)
(\ref{basicDR}) equivalently reads 
$
\widetilde F_{\mu\nu}=-\int d^4x\, \frac{\delta\widetilde  G_{\rho\sigma}}{\delta F_{\mu\nu}}\,G^{\rho\sigma}
$. Now writing $\widetilde F_{\mu\nu}=\frac{\delta}{\delta
  F_{\mu\nu}}\,\frac{1}{2}\!\int \!d^4x \:F_{\rho\sigma}\widetilde F^{\rho\sigma}$ we equivalently have
\eq
\frac{\delta}{\delta F_{\mu\nu}}\int \! d^4x\:F\widetilde F+G\widetilde G=0~, \label{NGZ1}
\en
where $F\widetilde F=F_{\rho\sigma}\widetilde F^{\rho\sigma}$ and similarly
for $G\widetilde G$.
We require this condition to hold for any field configuration $F$
(i.e. off shell of (\ref{max22}), (\ref{max11})) and
hence  we obtain the Noether-Gaillard-Zumino (NGZ) self-duality
condition\footnote{Note that (\ref{NGZ2}) (the integrated form of
  (\ref{NGZ})) also follows in a straightforward manner  by repeating
  the passages in \cite{GZ} but with $G$ 
  the functional derivative of the action rather than the partial
  derivative of the lagrangian \cite{KT, AFZ}. This makes a difference for nonlinear
  theories which also contain terms in derivatives of $F$.}
\eq
\int \! d^4x~F\widetilde F+G\widetilde G=0 \label{NGZ2}~.
\en
The vanishing of the integration constant is determined for example by
the condition $G=\HF$  for weak and slowly varying fields,
i.e. by the condition that in this regime the theory is
approximated by usual electromagnetism. 

\sk

We also observe that  the NGZ self-duality condition (\ref{NGZ2}) is equivalent to
the 
invariance of $S^{inv}=S-\frac{1}{4}\int \!d^4x\,F\widetilde G$,
indeed under a rotation (\ref{rotFG}) with infinitesimal parameter
$\al$ we have  
$S^{inv}[F']-S^{inv}[F]=-\frac{\al}{4}\int\!d^4x\; F\widetilde F+ G\widetilde G=0$.

\begin{Note}\label{note1}
If the Lagrangian $L(F)$ of the action functional $S[F]$ does not  depend on the derivatives
of $F$, then we cannot integrate by parts and condition (\ref{NGZ2}) is equivalent to 
\eq
F\widetilde F+G\widetilde G=0\label{NGZ}
\en
since the field configuration $F$ is arbitrary (and therefore with
arbitrary support in spacetime).
On shell of (\ref{max22}), (\ref{max11})  we can introduce the electric potential $A_\mu$ and the
magnetic one $B_\mu$ so that
$F_{\mu\nu}=\partial_\mu A_\nu-\partial_\nu A_\mu$, 
$G_{\mu\nu}=\partial_\mu B_\nu-\partial_\nu B_\mu$
and (\ref{NGZ}) becomes the (Noether-Gaillard-Zumino)  current conservation condition $\partial_\mu
J^\mu=\partial_\mu( A_\nu \widetilde F^{\mu\nu}+B_\nu\widetilde
G^{\mu\nu})=0$.

Examples of  theories satisfying (\ref{NGZ2}) and not
(\ref{NGZ}) are obtained in Appendix \ref{HDactions}, where we  generalize the example
presented in \cite{BN}.
\end{Note}

\begin{Note}\label{Note2}
If the Lagrangian $L(F)$ is in Minkowski spacetime and if it depends only on $F$ and not on its
derivatives, then Lorentz invariance implies that it depends on the
scalars $FF$ and $(F\widetilde F)^2$, where the square in $(F\widetilde F)^2$
is needed for parity symmetry (space inversion invariance). More in
general we can consider a Lagrangian in curved spacetime that depends only
on the (pseudo)scalars $FF$ and $F\HF$.
It is then possible to integrate the differential equation
(\ref{NGZ}): $F_{\mu\nu}\HF^{\mu\nu}-4\,{\big(}^{\!\!\!\!\!\!\ast}\;{\frac{\partial
    L}{\partial F^{\mu\nu}}}\big)\frac{\partial L}{\partial
  F_{\mu\nu}}=0$. 
The solution is presented in  \cite{GZS} (and an alternative form is
presented in \cite{HKS}, see also \cite{Ivanov:2003uj}), it depends
on an arbitrary real valued function $v(s)$ of a real variable $s$, with the
initial condition that in the limit $s\rightarrow 0$ then $v(s)\to
-s$.  
 However  $L(F)$ is explicitly determined only after inverting a
 function related to $v(s)$. Hence explicit solutions $L(F)$ 
in terms of simple functions are very difficult to be found.

This suggests to look for solutions $L(F)$, and more in general actions
 $S[F]$,  that are  power series in the
coupling constant $\la$. 
\end{Note}

\begin{Note}
Given an action $S[F]$ with self-dual
equations of motion the one parameter family of
theories defined by $S_c[F]=\frac{1}{c^2}S[cF]$ (with $c\geq0$, cf. end of Section \ref{AFEOM}) are also self-dual. 
This is so because for any given value of $c$ the action $S_c[F]$ satisfies the corresponding NGZ self-duality
conditions (\ref{NGZ2}):
\eq
\int\!d^4x~ F\widetilde F-2\frac{\delta S_c[F]}{\delta
 F}2\widetilde{\frac{\delta S_c[F]}{\delta F}}\,=0\label{FFgGG}~.
\en 
Indeed  $\frac{\delta S_c[F]}{\delta
  F}\widetilde{\frac{\delta S_c[F]}{\delta F}}=
\frac{1}{c^4}\frac{\delta S[cF]}{\delta
 F}\widetilde{\frac{\delta S[cF]}{\delta F}}=
\frac{1}{c^2}\frac{\delta S[cF]}{\delta
 cF}\widetilde{\frac{\delta S[cF]}{\delta cF}}$.
Therefore condition (\ref{FFgGG}) is equivalent to
$\int\!d^4x ~cF c\widetilde{F}-2\frac{\delta S[cF]}{\delta
 cF}2\widetilde{\frac{\delta S[cF]}{\delta cF}}\,=0$.
These are the self-duality conditions for the action $S[\hat F]$ with
$\hat F=cF$. Hence these conditions hold because
the self-duality conditions for the 
 initial action $S$ hold for any field configuration.

This result allows to provide jet another derivation of the invariance
under duality rotation of expression (\ref{Ttrace}) for self-dual actions:
One has just to recall that the variation of the action with respect to a duality invariant
parameter is duality invariant \cite{GZ}.

\end{Note}

\subsection{Complex and chiral variables}\label{complexvariables}

Following Schr\"odinger \cite{Schrodinger, GZS} it is convenient to consider the complex  variables 
\eq
T=F-iG~,~~\bar T=F+i G~,
\en
that under duality transform with a phase: $T\to  e^{-i\al}T$,
$\overline T \to e^{i\al}\overline T$, and that treat on an equal footing
the electric and magnetic field strengths $F$ and $G$.
In the new variables the NGZ self-duality condition (\ref{NGZ2}) reads $\int
d^4x \;\bar{T}\,{{\widetilde T}}=0$, or equivalently
\eq\label{GZN2}
\int \!d^4x \sqrt{g}\,~\bar{T}{{\HT}}\,=\,0~.
\en
Following \cite{BN} we further consider the complex (anti)selfdual
combinations 
$F^\pm=\frac{1}{2}(F\pm i\HF)$, $G^\pm=\frac{1}{2}(G\pm i\HG)$
and
\eqa\label{TPM}
T^+&=&\frac{1}{2}(T+i\HT)=F^+-iG^+~,~~~~~~~~~~T^-=\frac{1}{2}(T-i\HT)=F^--iG^-~,~~\\
\overline{T^+}&=&\frac{1}{2}(\bar T-i ~\,{\bar T}^{\!\!\!\!\!\!\!\!\!\ast\,~~})=F^-+iG^-={\Tbarm}~,~~
\overline{T^-} =\frac{1}{2}(\bar T+i\HbT)=F^+
+iG^+=\Tbarp\,.~~~~~~\label{OTPM}
\ena
The fields in the first row have duality charge
$+1$ because transform with $e^{-i\al}$ under the duality rotation (\ref{rotFG}), while
their complex conjugates in the second row have duality charge $-1$.
Complex conjugation inverts chirality hence $T^+$ and
$\overline{T^-}=\Tbarp$ have chirality $+1$ while 
$T^-$ and $\overline{T^+}=\Tbarm$ have chirality $-1$.

The (anti)selfdual combinations have definite behaviour in the
coupling constant $\la\to 0$ limit. Since in this limit we recover
usual electromagnetism we have $G\to\HF$
and $G^\pm\to \mp i F^\pm$, and hence
\eq
T^+\,\to\, 0~~,~~T^-\,\to 2 F^-~.
\en
The NGZ self-duality condition (\ref{NGZ2}) in these variables reads  (use
$({\HT})^{\, \pm}= {{(T^\pm)_{\,}}}^{\!\!\!\!\!\!\!\!\!\!\!\!\!\!\!\!\!\!\!\!\!\ast\,~~~}\,\,\,=\mp i T^\pm\,$)
\eq
\int \!d^4x\sqrt{g}\;\,\, T^+\,\overline{T^-} -\,\overline{T^+}\,T^-\, =0~.
\en
\subsection{The action functional ${\cal I}[T^-,{\overline{T^-}}]$\label{ITTm}}
As noticed in \cite{CKR},
 the Bossard and
Nicolai proposal \cite{BN} for constructing self-dual equations of
motions is easily expressed in terms of  chiral variables:
We  consider a real valued functional ${\cal I}[T^-,{\overline{T^-}}]$ in the 
chiral variables\footnote{We stress that 
the independent variables in ${\cal I}$ are $T^-$ and  its complex conjugate
  $\overline{T^-}$, just like in $S[F]$ or $S[F^-,F^+]$ the independent variables are
  $F^-$ and its complex conjugate $F^+$.  The variables $T^+,
  \overline{T^+}$ (and hence $T,\overline T$) are then defined
in terms of the $T^-$, $\overline{T^-}$ ones.} $T^- ,{\overline{T^-}}$
and define the constitutive relations (called deformed twisted self
duality conditions in \cite{BN}, and  nonlinear twisted self-duality
conditions in \cite{CKR})

\eq
{T^+}^{\mu\nu}=\frac{1}{\sqrt{g}}\frac{\delta {\cal I}[T^-,{\overline{T^-}}]}{\delta{\overline{T^-}_{\!\!\mu\nu}}}~~,~~~~{\overline{T^+}{_{}}^{\mu\nu}}=\frac{1}{\sqrt{g}}\frac{\delta {\cal I}[T^-,{\overline{T^-}}]}{\delta T^-_{\;\mu\nu}}\label{Iconst}~.
~
\en
Reality\footnote{The reality condition is ${\cal
    I}[T^-,{\overline{T^-}}]=\overline{{\cal
      I}[T^-,{\overline{T^-}}]}$. Then we  extend
${\cal I}[T^-,{\overline{T^-}}]$ to ${\widehat{\cal
  I}}[T^-,{\overline{U^-}}]\equiv \frac{1}{2}\big({\cal
  I}[T^-,{\overline{U^-}}]+\overline{{\cal
    I}[U^-,{\overline{T^-}}]}_{\,}\big)$ that by construction
satisfies 
$\overline{{\widehat{\cal I}}\big[T^-,{\overline{U^-}}\big]}={\widehat{\cal
  I}}\big[\,{{{U^-}}}\,,\, \overline{T^-}\,\big]$ for arbitrary
complex and independent fields $T^-$ and ${{U^-}}$. 
The functional variation in (\ref{Iconst}), where $\,{\overline{T^-}}$ is kept
independent from $T^-$, then explicitly reads  
$\,
T^+=\frac{1}{\sqrt{g}}\frac{\delta {\widehat{\cal I}}[T^-,{\overline{U^-}}]}{\delta{\overline{U^-}}}\Big|_{{U^-=T^-}}\,,~{\overline{T^+}}=\frac{1}{\sqrt{g}}\frac{\delta {\widehat{\cal I}}[T^-,{\overline{U^-}}]}{\delta T^-}\Big|_{{U^-=T^-}}
$.
}  of
$\cal I$
implies that the second equation is just the
complex conjugate of the first one, hence the constitutive relations
are 6 real equations as in (\ref{maxwww}) and in (\ref{Sconst}).
If  moreover ${\cal I}$ is duality invariant under $T^-\to
e^{-i\al}T^-$, $\overline{T^-}\to e^{i\al}\overline{T^-}$ then
relations (\ref{Iconst}) imply the NGZ self-duality condition (\ref{NGZ2});
indeed under an infinitesimal duality rotation $T^-\to T^-+\Delta T^-$,
$\Delta T^-=-i\al T^-$ 
we have:
\eq
0=\Delta{\cal I}=\int d^4x\;\,\frac{\delta{\cal I}}{\delta {\overline{T^-}}}\Delta
{\overline{T^-}}+
\frac{\delta{\cal I}}{\delta T^-}\Delta
T^-=i\al \int \!d^4x\sqrt{g}\,  \,\,T^+\,{\overline{T^-}} -\,{\overline{T^+}} \,T^-~.
\en
This is a  powerful approach because the constitutive relations are
easily given (though the dependence
$\widetilde G_{\mu\nu}=h_{\mu\nu}[F,\la]$ is determined implicitly), and the self-duality condition is also easily
implemented: just consider duality invariant functionals ${\cal I}$. Furthermore,
Lorentz (or, in curved spacetime, diffeomorphisms)  invariance of the
functional $\cal I$ implies Lorentz (diffeomorphisms)
covariance of the nonlinear and higher derivatives equations of motion. 
\sk\sk

The problem with this approach is that of finding an
action functional $S[F]$ such that the constitutive relations (\ref{Sconst})$\,$:
$
{G^{}_{}}^{\!\!\!\!\!\:\!\!\!\!\!\!\ast~~}{}^{\mu\nu}= \frac{2}{\sqrt{g}}\frac{\delta S[F]}{\delta F_{\mu\nu}} 
$,
are equivalent to the constitutive relations (\ref{Iconst}).

We first approach this problem perturbatively, and give explicit
expressions for the lowest order perturbations; in the next section
we solve the problem (albeit implicitly) by using a Legendre transform
between $S$ and ${\cal I}$.

In the perturbative approach we assume that ${\cal I}={\cal I}[T^-,\overline{T^-}]$ is a power
series in the coupling constant $\la$, 

\eq
{\cal I}[T^-,\overline{T^-}]={\cal I}^{[0]}[T^-,\overline{T^-}]+{\cal I}^{[1]}[T^-,\overline{T^-}]+{\cal I}^{[2]}[T^-,\overline{T^-}]+\ldots
\en
where ${\cal I}^{[n]}$ denotes the term proportional to $\la^n$, and
in this expansion  $T^-,\overline{T^-}$ are considered the elementary
independent fields (and hence $\la$ independent).

Then $S[F]=S[F^-,F^+]$ is found as a power series 
\eq
S[F^-,F^+]=S^{(0)}[F^-,F^+]+S^{(1)}[F^-,F^+]+S^{(2)}[F^-,F^+]+\ldots
\en
where  ${S}^{(n)}$ denotes the term proportional to $\la^n$, and
in this expansion  $F^-,F^+$ are the elementary  independent fields (and hence $\la$ independent).
The initial condition is  ${\cal I}^{[0]}=0$,  that corresponds to linear
electromagnetism, ${S}^{(0)}=-{\frac{1}{4}}\int\!d^4x\sqrt{g}\,F^2$.

Since $\overline{T^+}=F^-+iG^-$ implies ${\overline{T^+}}^{(n)}=i{G^-}^{(n)}$ for $n\geq
1$, we see that 
equivalence of the constitutive relations (\ref{Iconst}) and
(\ref{Sconst}), that we rewrite as $G^{\pm\,\mu\nu}=\frac{\pm2i}{\sqrt{g}}\frac{\delta S}{\delta
  F^\pm_{\,\mu\nu}}$, is  obtained by requiring order by order in $n$  that
the term  ${S}^{(n)}$
satisfies the condition
\eq\label{recS}
2\frac{\delta S^{(n)}}{\delta F^-_{\,\mu\nu}}=\Big(\,\frac{\delta {\cal I}}{\delta
  T^-_{\,\mu\nu}}\Big|_{{{{}^{T^-[F^-,F^+]}_{\overline{T^-}[F^-,F^+]}}}}\Big)^{(n)}
\en
where on the right hand side we consider 
$\frac{\delta {\cal I}}{\delta T^-}$ as a functional of $F^-$ and
$F^+$ because $T^-=T^-[F^-,F^+]$;
 the dependence
      $T^-=T^-[F^-,F^+]$
being implicitly determined by the chiral variables constitutive
relations (\ref{Iconst}) and the relations $T^{\pm}=F^\pm-iG^\pm$,
that, in order to stress that the independent variables are
$T^-$ and $\overline{T^-}$, we rewrite as
\eqa
&&~\,~~2F^-=T^-+\frac{1}{\sqrt{g}}\frac{\delta {\cal I}[T^-,{\overline{T^-}}]}{\delta T^-_{\;\mu\nu}}~~,\,~~~~~~~2F^+=\frac{1}{\sqrt{g}}\frac{\delta {\cal I}[T^-,{\overline{T^-}}]}{\delta \overline{T^-}_{\;\mu\nu}} +\overline{T^-}~,\label{al}\\~
&&-2iG^-=T^--\frac{1}{\sqrt{g}}\frac{\delta {\cal I}[T^-,{\overline{T^-}}]}{\delta T^-_{\;\mu\nu}}~~,~~~-2iG^+=\frac{1}{\sqrt{g}}\frac{\delta {\cal I}[T^-,{\overline{T^-}}]}{\delta \overline{T^-}_{\;\mu\nu}} -\overline{T^-}~.
\label{al1}
\ena
In Appendix \ref{app1} we determine the first two nontrivial terms of the
nonlinear and higher derivatives electromagnetic action associated
with an arbitrary functional ${\cal I}={\cal I}^{[0]}+{\cal I}^{[1]}+{\cal I}^{[2]}+\ldots$, with ${\cal I}^{[0]}=0$.
They read
\eqa
S^{(1)}[F^-,F^+]&=&\frac{1}{4}{\cal I}^{[1]}[2F^-,2F^+]~,~~\nn\\[.4em]
S^{(2)}[F^-,F^+]&=&\frac{1}{4}{\cal I}^{[2]}[2F^-,2F^+]-
\frac{1}{2}\int \!\!d^4x\frac{1}{\sqrt{g}}\,~\frac{\delta {S}^{(1)}}{\delta
F^-}
\frac{\delta {S}^{(1)}}{\delta
F^-}+\frac{\delta {S}^{(1)}}{\delta
F^+}\frac{\delta {S}^{(1)}}{\delta
F^+}~~.\label{Scorrec}
\ena
We recall that 
at zeroth 
order
$S^{(0)}[F^-,F^+]=-\frac{1}{4}\int\!d^4x\sqrt{g}~\, {F^-}^2 +_{\,} {F^+}^2=-\frac{1}{4}\int\!d^4x\sqrt{g}~ F^2$.

\subsection{From $S[F]$ to ${\cal I}[T^-,\overline{T^-}]$ via Legendre transform}
We now show, as in\cite{Ivanov:2003uj},  that ${\cal
  I}[T^-,\overline{T^-}]$ and $S[F]$ are related by 
\eq\label{LegendreT}
\!\frac{1}{4}{\cal I}[T^-,\overline{T^-}]=S[F]+\int\!d^4x\sqrt{g}\, ~
 \frac{1}{2}T^-F^--\frac{1}{8}{T^-}^2-\frac{1}{4}{F^-}^2 +
\frac{1}{2}\overline{T^-} F^+-\frac{1}{8}{\overline{T^-}}^2-\frac{1}{4}{F^+}^2~.
\en
This is actually a Legendre transform, and it implies that
the constitutive relations (\ref{Iconst}) are equivalent to the
constitutive relations 
(\ref{Sconst}),
i.e., $G^{\pm\,\mu\nu}=\frac{\pm2i}{\sqrt{g}}\frac{\delta S[F^-,F^+]}{\delta
  F^\pm_{\,\mu\nu}}$. 

In order to recognize (\ref{LegendreT}) as a Legendre transform we define the functional
\eq
U[F^-,F^+]=-2S[F^-,F^+]+\frac{1}{2}\int\! d^4x \sqrt{g}~{F^-}^2+{F^+}^2 ~.
\en
Recalling that $iG^-=F^--T^-$ (see (\ref{TPM})) the constitutive relations $G^{\pm\,\mu\nu}=\frac{\pm2i}{\sqrt{g}}\frac{\delta S[F^-,F^+]}{\delta
  F^\pm_{\,\mu\nu}}$ now read
\eq\label{TUF}
T^-=\frac{1}{\sqrt{g}}\frac{\delta U}{\delta F^-}~, ~~\overline{T^-}=\frac{1}{\sqrt{g}}\frac{\delta U}{\delta F^+}~.
\en
These relations (at least for weak and slowly varying fields) implicitly
determine $F^\pm=F^\pm[T^-,\overline{T^-}]$.
We then consider the Legendre transform
\eq\label{LTVU}
V[T^-,\overline{T^-}]=-U[F^-,{F^+}]+\int\!d^4x\sqrt{g}~\,T^-F^-+\overline{T^-}F^+~.
\en
Varying w.r.t. $T^-$ and $\overline{T^-}$ we obtain that the
dependence $F^\pm=F^\pm[T^-,\overline{T^-}]$ is given by
\eq\label{FVT}
F^-=\frac{1}{\sqrt{g}}\frac{\delta V}{\delta T^-}~,~~F^+=\frac{1}{\sqrt{g}}\frac{\delta V}{\delta {\overline{T^-}}}~.
\en
Therefore (\ref{FVT}) are the inverse relations of (\ref{TUF}), in
particular they are equivalent to 
 $G^{\pm\,\mu\nu}=\frac{\pm2i}{\sqrt{g}}\frac{\delta S[F^-,F^+]}{\delta
  F^\pm_{\,\mu\nu}}$. 
We now define the functional $ {\cal I}[T^-,\overline{T^-}]$ via
\eq
V[T^-,\overline{T^-}]=\frac{1}{2}{\cal
    I}[T^-,\overline{T^-}]+\frac{1}{4}\int\!d^4x\sqrt{g}~\,{T^-}^2+{\overline{T^-}}^2~.
\en
Relation (\ref{LegendreT}) is trivially equivalent to
(\ref{LTVU}). Furthermore the constitutive relations 
 $G^{\pm\,\mu\nu}=\frac{\pm2i}{\sqrt{g}}\frac{\delta S[F^-,F^+]}{\delta
  F^\pm_{\,\mu\nu}}$ and (\ref{Iconst}) are equivalent because
(\ref{FVT}) is easily seen to be equivalent to (\ref{al}), i.e., to  (\ref{Iconst}).

\sk
Let's now study duality rotations.
We consider $F$ to be the elementary fields and let
$S[F]$ give self-dual constitutive relations.
Under infinitesimal duality rotations  (\ref{rotFG}), $F\to F+\Delta F=F-\al G$,
$G\to G+\Delta G=G+\al F$ we have (since
$T^-=F^--\frac{2}{\sqrt{g}}\frac{\delta S}{\delta F^-}$) that $T^-\to
T^-+\Delta T^-=T^--i\al T^-$. We calculate the variation of  (\ref{LegendreT}) under duality
rotations. After a little algebra we
see that 
\eqa\label{sdequiv}
\Delta {\cal I}&=& {\cal I}[T^-+\Delta
T^-,\overline{T^-}+\Delta\overline{T^-}]- {\cal I}[T^-,\overline{T^-}]\\
&=&S[F+\Delta F]-S[F]+\frac{\al}{4}\int\!d^4x\sqrt{g}~\,G\widetilde
G-F\widetilde F
=
-\frac{\al}{4}\int\!d^4x\sqrt{g}~\,G\widetilde
G+F\widetilde F
=0\nn\ena
where we used that $S[F+\Delta F]-S[F]=\int\!d^4x\; \frac{\delta S}{\delta F}\Delta
F=-\frac{\al}{2}\int\!d^4x\;\widetilde G G$, and the  self-duality conditions (\ref{NGZ2}). 
Hence $\cal I$ 
is invariant under duality rotations.

Vice versa, we can consider $T^-$, $\overline{T^-}$ to be the elementary
fields and assume  ${\cal I}[T^-,\overline{T^-}]$ to be duality
invariant. Then from (\ref{al}) and $iG^-=F^--T^-$, i.e., form
(\ref{al}) and (\ref{al1}), it follows that
under the infinitesimal rotation $T^-\to T^-+\Delta T^-=T^--i\al T^-$ we have
 $F\to F+\Delta F=F-\al G$,
$G\to G+\Delta G=G+\al F$, and from (\ref{sdequiv}) we
 recover the self-duality conditions
(\ref{NGZ2}) for the action $S[F]$.
\sk
This shows the equivalence betweeen the $S[F]$ and the ${\cal
  I}[T^-,\overline{T^-}]$
formulations of self-dual constitutive relations. Hence
the deformed twisted self-duality condition
proposal originated in the context of supergravity counterterms is
actually the general framework needed to discuss self-dual theories
starting from a variational principle.

\section{Constitutive relations without self-duality\label{constitutiverelations}}

The constitutive relations  (\ref{maxwww}) or (\ref{Sconst}) express $G$ as a function
of $F$ and its derivatives. They do not treat on equal footing $F$
and $G$  and therefore their eventual compatibility with
duality symmetry is hidden.
%
On the other hand the independent chiral variables $T^-, \overline{T^-}$ of  the constitutive relations
(\ref{Iconst}) (the deformed twisted self duality
conditions) treat by construction $F$ and $G$  on equal footing, and
duality rotations are simply implemented via multiplication by a phase.
There however the relation betweeen $G$ and $F$ is implicitly
given.
Moreover, already the description of Born-Infeld theory is quite
nontrivial in these chiral variables.
We here further study  the nonlinear relations between these two
formulations and related ones. This study
sheds light on the structure of self-dual theories, in particular it will lead to a
closed form expression of  the constitutive relations (\ref{Iconst}) for the Born-Infeld theory.

We proceed with a manifestly duality symmetric reformulation of the constitutive relations
(\ref{maxwww}) (and more precisely of the relations (\ref{GNN}) below).
This is achieved doubling them (to $\HG=h[F,\la]$ and $\HF=k[G,\la]$)
and then constraining them via a symplectic
matrix $\cal M$. This matrix is well known in the study of  duality rotations in
linear electromagnetism coupled to scalar fields
  (see e.g. \cite{AFZ}). Here $\cal M$ will be in
general dependent  on the field strengths $F,G$ and their
derivatives, leading to nonlinear and higher derivatives
electromagnetism.
 Its structure will be fully determined only by requiring 
that  the doubled constitutive relations consistently give just 6 independent
equations that determine $G$ in terms of $F$ and vice versa. Notice that
even thought our aim is the study of self-dual
theories, in this section we do not assume that the constitutive relations
are compatible with duality symmetry.

The constraints on the $\cal M$ matrix  are then 
analized in terms of the  Schr\"odinger's variables  $T$, $\overline T$.
It is in these variables that Born-Infeld theory has an extemely
simple description \cite{Schrodinger, GZS}.

\subsection{The  ${\cal N}$ and $\M$ matrices}
More insights in the constitutive relations (\ref{maxwww})  can be obtained if we
restrict our study to the wide subclass that can be written as
\eq
{\HG}{}_{\mu\nu}={\cN_2}_{\,} F_{\mu\nu} +{\cN_1}_{\,} \HF{}_{\!\mu\nu}~,
\label{GNN}
\en
where $\cN_2$ is a real scalar field,
while $\cN_1$ is a real pseudo-scalar field (i.e., it is not invariant
under parity, or, if we are in curved spacetime, it is not invariant under
an orientation reversing coordinate transformation).  
Explicit examples of more general constitutive relations are in
 Appendix \ref{HDactions}.
 As usual in the literature we set 
\eq
\cN=\cN_1+i\cN_2~,
\en
then, relations
(\ref{GNN}) are equivalent to $G^+=\N F^+$.
In nonlinear theories $\cN$ depends on the field strength $F$, and in
higher derivative theories also on derivatives of $F$, we have
therefore in general a functional dependence $\cN=\cN[F,\la]$.
Furthermore $\cN$ is required to satisfy $\cN\to -i$ in
the limit $\la\to 0$ so that we recover classical electromagnetism
when the coupling constant $\la\to 0$, or otherwise stated, in the
weak and slowly varying field limit, i.e., when we discard 
higher powers of $F$ and derivatives of $F$.
We also assume that $\cN$ can be expanded in power series of the
coupling constant\footnote{By $\la$ we can denote also more than
one coupling constant. For example when a nonlinear theory in flat
space is generalized to a curved background there naturally appears
a new coupling  related to the
background curvature. Similarly, as already said,  if the theory has higher derivatives
so that it can be expanded in appropriate powers of derivatives of
$F$.} $\la$ (we will relax this assumption in Note \ref{Note5}). Then, since $\cN_2= -1+O(\la)$,
$\N_2$ is invertible, and from relation (\ref{GNN}) we obtain $\widetilde F=\cN_2^{-1} \cN_1
F-\cN_2^{-1} G$ and $\widetilde G=\cN_2
F+\cN_1\cN_2^{-1}\cN_1F-\cN_1\cN_2^{-1} G$
so that the constitutive relation \eqn{GNN} is equivalent to the more duality
symmetric one
\eq\label{FFomMFF}
\left(\begin{array}{c}
\HF \\
\HG
\end{array}\right)=
\left(\begin{array}{cc}
0 & -1\\
1 & 0
\end{array}\right)\,\cM\,
\left(\begin{array}{c} 
F  \\
G
\end{array}\right)
\en
where the matrix $\cM$ is given by
\eq
\M(\N)=
\left(\begin{array}{cc}
1 & -\N_1\\
0 & 1
\end{array}\right)
\left(\begin{array}{cc}
\N_2 & 0  \\
0 & \N_2^{-1}
\end{array}\right)
\left(\begin{array}{cc}
1 & 0\\
- \N_1 & 1
\end{array}\right)
=
\left(\begin{array}{cc}
\N_2 +\N_1 \,\N_2^{-1}\,  \N_1 &~ - \N_1 \,\N_2^{-1}  \\
-\N_2^{-1}\, \N_1 &~ \N_2^{-1}
\end{array}\right)~.~~
\label{M(N)}
\en
Finally, in order to really treat on equal footing the electric and
magnetic field strengths $F$ and $G$, we should consider
functionals ${N}_1[F,G,\la]$ and ${N}_2[F,G,\la]$ such that the constitutive relations
${\HG}={N_2[F,G,\la]}_{\,} F +{N_1[F,G,\la]}_{{\,}} \HF$
are equivalent to (\ref{GNN}), i.e.,  such that on shell of these relations, ${N}_1[F,G,\la]=\N_1[F,\la]$ and 
${N}_2[F,G,\la]=\N_2[F,\la]$.

Since we assume  $N_1[F,G,\la]$ and $N_2[F,G,\la]$ to be  power series
in $\la$ with $N_1=O(\la)$ and 
$N_2=-1+ O(\la)$ 
 the constitutive relations
${\HG}={N_2[F,G,\la]}_{\,} F +{N_1[F,G,\la]}_{{}} \HF$ are well given in the sense that they are always equivalent to the 
${\HG}={\N_2[F,\la]}_{\,} F +{\N_1[F,\la]}_{{}} \HF$ ones
  (just expand in power series of $\la$ and iteratively substitute $G$ in
  $N_1[F,G,\la]$ and $N_2[F,G,\la]$).

Henceforth, with slight abuse of notation,
from now on the $\N$, $\N_1$, $\N_2$ fields in
(\ref{GNN})-(\ref{M(N)}) will  in general be functionals of
both $F$ and $G$.

\sk

The matrix $\M(\N)$ in (\ref{M(N)}) is symmetric and symplectic 
(it has indeed determinant equal to 1). The space of symmetric and symplectic
matrices has two disconnected components, that of  positive definite
and of negative definite matrices.  $\M(\N)$
is negative definite because $\N_2^{-1}\to -1+O(\la)$.
Recalling that any symmetric, symplectic and negative definite $2\times 2$ matrix is of the kind
(\ref{M(N)}) with $\N_1$ real and $\N_2$ real and negative
(for a proof see for example the review \cite{AFZ}, Appendix A, where
it is also shown  that $\M$ and $\N=\N_1+i\N_2$ parametrize the coset
space $Sp(2, \mathbb{R})/U(1)$), we have that
\sk
\begin{Proposition}\label{propos2} Any symmetric and symplectic $2\times 2$
matrix $\M$
that has a power series expansion in $\la$ with $\M=-1+O(\la)$ is of the kind 
(\ref{M(N)}) with $\N_1=O(\la)$ real and $\N_2=-1+O(\la)$ real. 
\end{Proposition}

We now reverse the argument that led from (\ref{GNN}) to (\ref{FFomMFF}).
 We consider 
constitutive relations of the form
\eq\label{FFomMFFp}
\left(\begin{array}{c}
\HF \\
\HG
\end{array}\right)=
\left(\begin{array}{cc}
0 & -1\\
1 & 0
\end{array}\right)\,\cM[F,G,\la]\,
\left(\begin{array}{c} 
F  \\
G
\end{array}\right)
\en
that treat on equal footing $F$ and $G$, and where $\M=\M[F,G,\la]$ is now
an {\sl arbitrary} real
$2\times 2$ matrix  (with scalar entries $\M_{ij}$).
We require
$\M=-1+O(\la)$
so that we recover classical
electromagnetism when the coupling constant $\la\to 0$.
A priory (\ref{FFomMFFp}) is a set of  12 real equations, twice as
much as those present in the constitutive relations (\ref{GNN}). We want only 6 of
these 12 relations to be independent so to be able to determine $G$
in terms of  independent fields $F$ (or equivalently $F$ in terms of  independent fields $G$). Only in this case the constitutive relations are well given. 
\sk
\begin{Proposition}\label{propos3} The constitutive relations  (\ref{FFomMFFp}) 
with $\M[F,G,\la]=-1+O(\la)$ are  well given  if and only if on shell of  (\ref{FFomMFFp})  the matrix 
$\M[F,G,\la]$
is symmetric and symplectic.
\end{Proposition}

\begin{proof}
i) Let $\M[F,G,\la]=-1+O(\la)$ be symmetric and
symplectic on shell of  (\ref{FFomMFFp}).
Then, because of Proposition \ref{propos2}, there exists a unique $\N[F,G,\la]=-i+O(\la)$ such that $\M[F,G,\la]=\M(\N)$ on
shell of  (\ref{FFomMFFp}). Hence (\ref{FFomMFFp}) is 
equivalent to  (\ref{GNN}) and therefore gives well defined
constitutive relations. 
\\
ii) If the constitutive relations  (\ref{FFomMFFp})
are a set of 6 independent relations that determine $G$ in terms of
$F$ then the matrix entry $\M_{22}\not=0$ (because otherwise from  (\ref{FFomMFFp}),
we would have $\HF=-\M_{21}F$ that constraints the independent
fields $F$). 
It follows that  (\ref{FFomMFFp}) is equivalent to
$G=-\M_{22}^{-1}\M_{21} F-\M_{22}^{-1}\HF$, i.e.  to
$\HG=\M_{22}^{-1} F -\M_{22}^{-1}\M_{21} \HF$.
Repeating the argument that lead from (\ref{GNN}) to (\ref{FFomMFF})  
we conclude that  (\ref{FFomMFFp})  is equivalent to the 
equations  
\eq\label{FFomMFFpp}
\left(\begin{array}{c}
\HF \\
\HG
\end{array}\right)=
\left(\begin{array}{cc}
0 & -1\\
1 & 0
\end{array}\right)
\left(\begin{array}{cc}
\M_{22}^{-1}+\M_{22}^{-1}\M_{21}^2 \:& \M_{21}\\
\M_{21} & \M_{22}
\end{array}\right)
\left(\begin{array}{c} 
F  \\
G
\end{array}\right)~.
\en
We show that on shell of the relations  (\ref{FFomMFFp})  the matrix
$\M[F,G,\la]$ is symmetric and symplectic because
\eq\label{MSYMM}
\M[F,G,\la]=
\small{\left(\begin{array}{cc}
\M_{22}^{-1}+\M_{22}^{-1}\M_{21}^2 \:& \M_{21}\\
\M_{21} &\ M_{22}
\end{array}\right)}~~~~\mbox{ {\sl (on shell)}}.
\en
Since by hypothesis the relations (\ref{FFomMFFp})  determine $G$ in terms of $F$ and $\HG= - F +O(\la)$,
we can also determine $F$ in terms of $G$ as a power series in $\la$.
Then (\ref{FFomMFFp})  is also equivalent to $\HG=\M_{11} F+\M_{12}
G$ and, observing that independence of the $G$ fields implies that the matrix
entry $\M_{11}\not=0$,
we conclude that  (\ref{FFomMFFp})  is as well  equivalent to
$F=\M^{-1}_{11}\HG -\M^{-1}_{11}\M_{12}G$, that we rewrite as 
\eq\label{FpG0}
F^+=P\,G^+~,~~P\equiv (-\M^{-1}_{11}-i\M^{-1}_{11}\M_{12})~.
\en
Similarly  (\ref{FFomMFFpp})  is also equivalent to its second row, 
$\HG=(\M^{-1}_{22}+\M_{22}^{-1}\M_{21}^2) F+\M_{21} G$ that we
rewrite as
\eq\label{FpG1}
F^+=Q\,G^+~,~~Q\equiv \Big(-(M^{-1}_{22}+M_{22}^{-1}M_{21}^2)^{-1}-i(M^{-1}_{22}+M_{22}^{-1}M_{21}^2)^{-1}M_{21}\Big) ~.
\en
Independence of the fields $G^+$ implies that subtracting (\ref{FpG1}) to (\ref{FpG0})  we obtain that $P-Q=0$ in
each region of spacetime where $G^+\not=0$. Moreover $P-Q=0$ in
each region of spacetime where $G^+=0$ because $G^+=0$ in that region implies
$P=1$ and $Q=1$ in that same region (we consider $\M[F,G,\la]$ to be a
local functional of $F$ and $G$). This shows the on shell equality $P=Q$.
Then equality (\ref{MSYMM}) immediately follows.
\end{proof}

\sk

\begin{Note} \label{Note5}
We have assumed that the constitutive relations can 
be written as power series expansions in $\la$. We here relax this 
assumption and consider constitutive relations (\ref{GNN}) such that
$\N[F,G,\la]=-i$ (or $\M[F,G,\la]=-1$) 
for the field configuration $F=G=0$;
this is equivalent to state that for weak and slowly varying fields
$\HG\approx -F$ (i.e., that in this regime the constitutive relations
are those of usual electromagnetism). Then
applying the implicit function theorem to the constitutive relations (\ref{GNN}) 
we know that  there exists  neighbourhoods of the field
configurations $F=0$, $G=0$ such that (\ref{GNN}) are equivalent to
the explicit expressions $G=G[F,\la]$ and  $F=F[G,\la]$.
The result of this section therefore holds also without the power series expansion
in $\la$ assumption: just consider fields sufficiently weak and slowly varying.
\end{Note}
\subsection{Complex variables}
As in Section \ref {complexvariables} it is fruitful to consider the
complex variables $T=F-iG$,  $\bar{T}=F+iG$.
The transition from the real to the complex variables is given by the
symplectic and unitary matrix $\A^t$
where
\eq
{\cal A}={1\over \sqrt{2}}\left(
\begin{array}{cc}
1  & 1\\
-i & i
\end{array}\right)~~,~~~~{\cal A}^{-1}={\cal A}^\dagger~. \label{defAAm1}
\en
The equation of motions in these variables read $dT=0$, with constitutive
relations obtained applying the matrix ${\cal A}^t$
to (\ref{FFomMFFp}):
\eq\label{TRT}
\left(\begin{array}{c}
\HT \\
\HbT
\end{array}\right)=
-i\left(\begin{array}{cc}
1 & \,0\\
0 & -1
\end{array}\right)\, {\cal A}^t \M \overline{\cal A}
\left(\begin{array}{c} 
T  \\
\bar T
\end{array}\right)~,
\en
where $ {\cal A}^t \M \overline{\cal A}$, on shell of (\ref{TRT}),  is
 complex symplectic and pseudounitary  w.r.t the metric 
$\big({}^1_0{}^{~0}_{-1}\big)$, i.e. it belongs to $Sp(2,\cc)\cap
U(1,1)=SU(1,1)$. It is also Hermitian and negative definite.
These properties uniquely characterize the matrices $ {\cal A}^t \M \overline{\cal A}$ 
as the  matrices 
\eq
\left(\begin{array}{cc}
 -\sqrt{1+\tau\bar \tau^{}} & \, -i \tau\\
i \bar \tau & -\sqrt{1+\tau\bar \tau^{}}
\end{array}\right)
\en
where $\tau=\tau[T,\bar T]$ is a complex field that depends on $T$,
$\bar T$ and possibly also their derivatives. 
We then see that the constitutive relations (\ref{TRT}) are equivalent
to the equations
\eq\label{TccT}
\HT{}_{\!\mu\nu}=i \sqrt {1+\tau\bar \tau} \,T_{\mu\nu} -\tau\,\bar T_{\!\mu\nu}
\en
Notice that if ${\cal
M}=-1+O(\la)$ (or equivalently ${\cal N}=-i+O(\la)$), then
$\tau=O(\la)$. In particular electromagnetism is obtained setting $\tau=0$.

\sk
In conclusion equations (\ref{TccT}) 
are the most general way of writing six independent real equations
that allow to express $G=\frac{i}{2}(T+\bar T)$ in terms of
$F=\frac{1}{2}(T+\bar T)$ as in (\ref{GNN})
(equivalently $F$ in terms of $G$).
These constitutive relations
 are determined by a
 complex function ${\cal N}$ (depending on $F,G$ and their derivatives
 ${\cal N}={\cal N}[F,G]$) or
 equivalently $\tau$
(depending on $T, \bar T$ and their derivatives $\tau=\tau[T,\bar T]$).

\section{Schr\"odinger approach to self-duality conditions\label{SCH}}
In the previous section we have clarified the structure of the
constitutive relations for an arbitrary nonlinear theory of
electromagnetism. The theory can also be with higher
derivatives of the field strength because the complex
field ${\cal N}$, or equivalently the matrix $\cal M$ in
(\ref{FFomMFFp}) of (pseudo)scalar entries, can depend also on
derivatives of the electric and magnetic field strengths $F$ and $G$.

We now further examine the constitutive relations for theories that satisfy the
NGZ self-duality condition (\ref{NGZ}), i.e., $\overline T\widetilde T=0$,
or equivalently,
\eq
\overline T\HT=0~.\label{TGZ}
\en
The constitutive relations (\ref{TccT}) 
determine the dependence of the magnetic field strength $G$ form the
electric one $F$ or vice versa. We notice that this dependence is determined also if
we constrain the fields in (\ref{TccT}) to satisfy the condition
$T\HT\not=0$. This is so because the set of field
configurations satisfying
$T\widetilde T\not=0$ is dense in the set of unconstrained field configurations.
Hence  if we multiply or divide  the constitutive relations (\ref{TccT})
by ${T\widetilde T}$ we obtain a set of equivalent constitutive
relations.
Having explained why we can freely divide by $T\widetilde T$ we can state the
following  
\sk
\begin{Proposition}\label{propos4} The constitutive relations  (\ref{TccT})
and the self-duality conditions (\ref{TGZ}) are equivalent to 
defining a nonlinear and higher derivatives extension of usual
electromagnetism by the relations
\eq\label{1/TT}
\HT{}_{\!\mu\nu}=-\frac{T^2}{T \HT} T_{\mu\nu}-\tau \bar T_{\mu\nu}~,
\en
that henceforth we call self-dual constitutive relations in
Schr\"odinger  variables.

Equivalently we have the  self-dual constitutive
relations
\eq\label{sfce1}
\HT{}_{\!\mu\nu}=-\frac{T^2}{T \HT } T_{\mu\nu}-\frac{T\bar
  T}{\overline{T\HT}} \bar T_{\mu\nu}~,
\en
\eq\label{sfce2}
{T\bar T} =r \,| T\HT|
\en
where the second equation is a scalar equation where $| T\HT|$ is
the modulus of  $T\HT$  and $r$ is a
dimensionless scalar field that depends on $T,\bar T$ and their
derivatives, that takes values in the non-negative real number and
that is duality invariant.
\end{Proposition}
\begin{proof}
Contracting the indices of (\ref{TccT}) with $\HT{}_{\mu\nu}$ we obtain
\eq
-T^2=i\sqrt{1+\tau\bar \tau} \,T\HT ~.\label{T2sq}
\en
 Hence  the self-duality condition (\ref{TGZ}) (i.e. (\ref{NGZ})), and the constitutive
 relations (\ref{TccT})  imply (\ref{1/TT}).
\sk
Vice versa (\ref{1/TT}) implies   (\ref{TGZ})  and (\ref{TccT}). Indeed, contracting (\ref{1/TT}) with $\HT{}_{\!\mu\nu}$ we
obtain $\overline T \HT=0$. This is trivially the case if $\tau\not=0$.
It holds also if $\tau=0$ because then  (\ref{1/TT}) reads 
$
\HT=-\frac{T^2}{T\,\,\,T_{}^{{}^{\!\!\!\!\!\!\;\!\!\!\ast}~}} T
$
, i.e.,
$
(T\HT)^2=-T^2T^2
$
that implies 
$T=\pm i \HT$, i.e., $F=\pm\HG$. This last relation implies the
self-duality condition  (\ref{TGZ}).

In order to show that  (\ref{1/TT}) implies  (\ref{TccT}) first we contract 
 (\ref{1/TT}) with
$\HbT{}_{\mu\nu}$, and obtain
\eq\label{tauTT}
{T\bar T}=\tau\, \overline{T \HT}~.
\en
Then we contract (\ref{1/TT}) with $T_{\mu\nu}$, and  obtain
\eq
T\HT=-\frac{T^2 }{T\HT} T^2 -\tau T\overline T~.\en
This expression and the complex conjugate of (\ref{tauTT}) imply
$1+\tau\overline \tau=-\frac{T^2 T^2}{{(T\,\,\,T_{}^{{}^{\!\!\!\!\!\!\;\!\!\!\ast}~})}^2}$, and hence
$-\frac{T^2}{T\,\,\,T_{}^{{}^{\!\!\!\!\!\!\;\!\!\!\ast}~}}=i\sqrt{1+\tau\overline \tau}$, that
substituted in (\ref{tauTT}) gives (\ref{TccT}), as was to be
proven. The sign of
the square root $\sqrt{1+\tau\overline \tau}$ is determined considering  the limit $\la\to 0$,
where we want to recover usual electromagnetism, that in these variables
reads $\HT=i T$.
\sk

The self-duality condition $\overline T
\HT=0$ implies  (\ref{tauTT}) that fixes the phase of $\tau$ to equal  that
of $T\HT$. This
constraint is automatically satisfied by setting $r=|\tau|$ and
\eq\label{defr}
\tau= r\, \frac{T\HT}{|T\HT|}~.
\en
The equivalence of (\ref{1/TT}) with the self-dual constitutive relations
(\ref{sfce1}), (\ref{sfce2}) is then immediate.
Trivially $r\geq 0$. Finally, recalling that $F$ and $\HG$ are
tensors while $\HF$ and $G$ are pseudo-tensors  we easily
check that $T\overline T$ and $T\HT\;\overline{T\HT}$ are  scalars, hence
$r$ is a scalar field depending on $T,\overline T$ and their
derivatives (i.e., $r$ is invariant under orientation reversal).
Duality invariance of $r$ (under $T\to e^{-i\al} T$) immediately follows from 
(\ref{sfce2}).
\end{proof}
In the self-duality conditions (\ref{sfce1}),  (\ref{sfce2}) we have
been able to disentangle the general relations that a self-dual
theory must satisfy, i.e.,  (\ref{sfce1}), from the specific condition
that defines the nonlinear theory: the scalar equation  (\ref{sfce2})
that determines the ratio $T\overline T/|T\HT|$.
Nonlinear self-dual theories are defined by imposing that this ratio equals an
arbitrary duality invariant real and nonnegative 
scalar function $r$ of $T,\overline T$ and their derivatives. 
\begin{Example}\label{EEEEX}
{\sl Linear electromagnetism} ($G=\HF$) corresponds to the case $r=0$. Indeed
$T\overline T=0$ in linear electromagnetism, while $T\HT=2F\HF+2iF^2$ is
arbitrary. \\
{\sl Born-Infeld nonlinear theory} satisfies the constitutive
relations, $\la\;{\bar
  T}^{\!\!\!\!\!\!\!\!\!\ast\,~~}{}^{\!\mu\nu}=\frac{\partial}{\partial
T_{\mu\nu}}\big(\frac{4\,T^2}{T\,\,\,T_{}^{{}^{\!\!\!\!\!\!\;\!\!\!\ast}~}}\big)$, i.e.,  
\eq
\HT{}_{\!\mu\nu}=-\frac{T^2}{T\HT} T_{\mu\nu}-\frac{\la}{8}(T\HT)\,\overline T_{\mu\nu}\label{BIconst}
\en
as remarked by Schr\"odinger \cite{Schrodinger}, see \cite{GZS} for a clear
account in nowadays notations. Comparison with (\ref{sfce1}) and (\ref{sfce2}), 
shows that  Born-Infeld theory is determined by
\eq
r=\frac{\la}{8} |T\HT|~.\label{rBI}
\en
\end{Example}
\sk
We gain further insights in the self-dual constitutive relations by
analyzing the phases and moduli of the scalars fields that enter
(\ref{sfce1}) and (\ref{sfce2}).
Relation (\ref{T2sq}) implies that the
phase of $T\HT$ is bigger than the phase of $T^2$ by a
$\pi/2$ angle. In polar coordinates we have,
 \eq\label{fff}
T^2=|T^2| e^{i\varphi}~,~~T\HT =i|T\HT| e^{i\varphi}~.
\en
Use of (\ref{tauTT}) leads to the relation
$\tau\overline\tau=
{|T\overline T|^2}/{|T\HT|^2}$, that inserted in  (\ref{T2sq}) gives\footnote{This equation suggests to set $\,\cosh\beta={\rho_{T^2}}/{\rho^{}_{{T^{\,\ast\,}{\!}T}}}\,$,
$\,\sinh \beta ={\rho_{T\overline T}}/{\rho^{}_{{T^{\,\ast\,}{\!}T}}}\,$, 
so that (\ref{rrr}) is automatically satisfied. With these
 variables the constitutive relations read
$\,
\HT=i \cosh \beta\, T-i\, \sinh\beta
\frac{T^2}{\rho_{T^2}}\, \overline T
$.
Different nonlinear theories are determined by the dependence of the
angle $\beta$ from the fields $T,\overline T$ and their derivatives.} 
\eq
|{T^2}|^2=|{T\HT}|^2+|{T\overline T}|^2~.
\label{rrr}
\en
\subsection{Chiral variables\label{cfr}}

The self-dual constitutive relations further simplify when we rewrite them
in term of the chiral variables $T^+,T^-$ and their complex conjugates.

We consider the Hodge dual of equation (\ref{sfce1}), sum it to $\pm i$-times
(\ref{sfce1}), and, with the help of (\ref{TPM}) and (\ref{OTPM}), we
obtain the equivalent relations
\eq
T^\pm_{\;\mu\nu}=-\frac{T\overline T}{{2 T^\mp}^2}\,
\frac{T\HT}{\overline {T\HT}} \:{\overline {T^\mp}}^{}_{\!\!\mu\nu}\label{tpmtbpm}
\en
where $2{{T}^{\,\mp}}^2=T^2\mp i T\HT=(|T^2|\pm |{T\HT}|)e^{i\varphi} $. Further 
use of the phase relations (\ref{fff}) leads to 
$T^\pm_{\,\mu\nu}=\frac{T\overline
  T}{{2\overline{T^\mp}^2}} \,\overline{T^\mp}^{}_{\!\!\mu\nu}$, i.e., to
\eq
T^+_{\;\mu\nu}=t \,e^{i\varphi}\: \overline{T^-}^{}_{\!\!\mu\nu}~,\label{tp}
\en
and $
T^-_{\;\mu\nu}=t^-  e^{i\varphi}\:{\overline{T^+}}^{}_{\!\!\mu\nu}$,
where the dimensionless,  nonnegative and duality rotation invariant scalar fields $t$ and $t^-$ are defined by
\eq
t\equiv \frac{T\overline T}{|T^2| + |T\HT|}~,\label{ttt}
\en
and $t^-\equiv \frac{T\overline T}{|T^2| - |T\HT|}$.
Equations $T^-_{\,\mu\nu}=t^- e^{i\varphi}\:{\overline {T^+}}^{}_{\!\!\mu\nu}$ are equivalent
to  $T^+_{\,\mu\nu}=t e^{i\varphi}\:{\overline {T^-}}^{}_{\!\!\mu\nu}$ because, due to
(\ref{rrr}), $t^-=t^{-1}$.
\sk

The scalar equation  (\ref{sfce2}) determines the value of the ratio
$T\overline T/|T\HT|$. Because of the moduli relation
(\ref{rrr}), it equivalently determines the ratio $t$ in (\ref{ttt}). 
Therefore, as in the previous section (see paragraph after the proof of Proposition \ref{propos4}),  we can conclude that 
(\ref{tp}) is the general relation that a self-dual theory must
satisfy, while  the specific condition that
defines the nonlinear theory is the dependence of the real
nonnegative duality invariant scalar function $t$ from a set of independent variables
and their derivatives, for example $T^-$ and $\overline{T^-}$.
\sk

It is useful to present the explicit relation between the ratio $r=T\overline T/|T\HT|$
and $t$.
 We calculate 
\eq|{T^-}^2|(1-t^2)=\frac{1}{2}(|T^2|+|T\HT|)(1-t^2)
=|T\HT|\label{useful}~,
\en 
multiply this last equality by $r=T\overline
T/|T\HT|$ and obtain
\eq
(1-t^2)r=2t~.\label{trrel}
\en

\section{Nonlinear theories without higher derivatives\label{nhd}}

If the constitutive relations
$G_{\mu\nu}=h_{\mu\nu}[F,\la]$ (see (\ref{maxwww})) do not involve derivatives of the fields
then, as noticed in the introduction,  any antisymmetric 2-tensor is a linear combination of
$F_{\mu\nu}$ and $\HF{}_{\!\mu\nu}$ with coefficients that are
(pseudo)scalar functions of  $F_{\mu\nu}$. Hence the constitutive
relations (\ref{GNN}) or (\ref{FFomMFF}) are the most general 
ones. Furthermore, if we are in  Minkowski spacetime
Lorentz invariance implies that 
the field $\cal N$ in
(\ref{GNN}) and the matrix $\cal M$ in (\ref{FFomMFF}) can be expressed in terms of  
the Lorentz invariant combinations $F^2$ and $(F\HF)$.
Similarly, if we choose the chiral fields $T^-$ and $\overline{T^-}$ as independent
variables (cf. Sections \ref{ITTm} and \ref{cfr}) then any Lorentz
invariant field is a function of ${T^-}^2$ and ${\overline{T^-}}^{\,2}$ . 

More in
general we consider theories in curved spacetime that depend only on
$F^2$ and $F\HF$, or ${T^-}^{2\!}$ and $\overline{T^-}^{\,2}$.
 Since the
action functional ${\cal I}[T^-,\overline{T^-}]$ studied in Section
\ref{ITTm} and the 
scalar field $t$ defined in (\ref{ttt}) are duality invariant, and under a duality of angle $\al$ we have the
phase rotation  ${T^-}^2\to e^{2i\al}{T^-}^2$, we conclude that ${\cal
I}$ and $t$ depend only on the modulus of ${T^-}^2$, hence 
${\cal I}={\cal I}[T^-,\overline{T^-}]$ and $t=t[T^-,\overline{T^-}]$ simplify to
\eq
{\cal I}=\frac{1}{\la}\int\!d^4x\sqrt{g}\:{I}(\uuu) ~,~~t=t(\uuu)~,
\en
where $I(\uuu)$ is an adimensional scalar function, and the variable
$\uuu$ is defined by

\eq
\uuu_{\,}\equiv_{\,} 2\la|{T^-}^2|_{\,}=_{\,}\la(|T^2|+|T\HT|)~.
\en

\noi Similarly, the constitutive relations (\ref{Iconst}) simplify to 
\eq
{T^+}^{\mu\nu}=
\frac{1}{\la}\frac{\partial I}{\partial {\overline{T^-}_{\!\!\mu\nu}}}=
\frac{1}{\la}\frac{d I}{d \uuu}_{\,}\frac{\partial \uuu\,}{\partial {\overline{T^-}_{\!\!\mu\nu}}}~,
\en and comparison with  with (\ref{tp}) leads to
\eq
t={2} \frac{d {I}}{d \uuu}~.\label{Ioft}
\en
In fact, deriving $\uuu^2$ we obtain $\frac{\partial \uuu}{\partial
\overline{T^-}_{\!\!\!\mu\nu}}=2\la e^{i\varphi } \,\overline{T^-}_{\!\!\!\mu\nu}$ where
we used the same conventions as in footnote \ref{funo}, and that
${T^-}^2=|{T^-}^2|e^{i\varphi}$ (see expression immediately after
(\ref{tpmtbpm})). 

\subsection{Born-Infeld nonlinear theory}
In this section we determine the scalar field $t=t(\uuu)=2\frac{d I}{d u}$ in case of
Born-Infeld theory. This is doable thanks to Schr\"odinger's
formulation (\ref{BIconst}) of Born-Infeld theory, that explicitly
gives $r=\frac{\la}{8}|T\HT|$, see (\ref{rBI}). 
Then from (\ref{useful}) we have 
\eq
r=\frac{1}{16}\uuu(1-t^2)~,
\en
and recalling (\ref{trrel}) we obtain \cite{AFunp, IZ}
\eq
(1-t^2)^2 \uuu= 32 t~. \label{poleq4}
\en
Now in the limit $\uuu\to 0$, i.e., $\la\to 0$, we see from (\ref{ttt}) that 
$t\to 0$. 
The function $t=t(\uuu)$ defining Born-Infeld theory
 is therefore given by the unique positive root of the fourth order
polynomial equation (\ref{poleq4}) that has the correct $\la\to 0$
limit. Explicitly,
\eq
t=\frac{1}{\sqrt{3}}\Big(\sqrt{1+s+s^{-1}} - \sqrt{2-s-s^{-1}
+\frac{24\sqrt{3}}{\uuu\sqrt{1+s+s^{-1}}}}~\,\Big) ~,\label{radical1}
\en
where
\eq
s=\frac{1}{\uuu} \Big(216_{\,} \uuu +12 \sqrt{3}\sqrt{108+\uuu^2}_{\,} \uuu+ \uuu^3\Big)^{\mbox{$\frac{1}{3}$}}~.
\label{radical2}\en

\subsection{The hypergeometric function and its hidden identity}
In \cite{CKR} the action functional $\cal I$ and the function $t(\uuu)$
corresponding to the Born-Infeld action  were found via
an iterative procedure order by order in $\la$ (or equivalently in $\uuu$).  The first
coefficients of the power series expansion of $t(\uuu)$ were recognized to
be those of a generalized hypergeometric function, leading to the conclusion
\eq
t(\uuu)=\frac{\uuu}{32} {\,}_3F_2\Big(\frac{1}{2}, \frac{3}{4}, \frac{5}{4};\,
  \frac{4}{3}, \frac{5}{3};\,-\frac{ \uuu^2}{3^3\cdot 2^2}\Big)~,\label{3F2}
\en
and, integrating (\ref{Ioft}),
\eq
{I}(\uuu)={6}\left(1-
 {\,}_3F_2\Big(-\frac{1}{2}, -\frac{1}{4}, \frac{1}{4},\,
  \frac{1}{3}, \frac{2}{3};\,-\frac{\uuu^2}{3^3\cdot 2^2}\Big)\right)~.
\en
We have checked that the expansion in power series of $\uuu$  of the closed form expression of
$t(\uuu)$ derived in (\ref{radical1}),(\ref{radical2}) coincides,  
up to order $O(\uuu^{1000})$  with $\frac{\uuu}{32}$ times the hypergeometric function  in (\ref{3F2}).
 Therefore we conjecture that the hypergeometric function in
 (\ref{3F2})
\eq
{\mathfrak F}(\uuu^2)\equiv{}_3F_2\Big(\frac{1}{2}, \frac{3}{4}, \frac{5}{4};\,
  \frac{4}{3}, \frac{5}{3};\,-\frac{\uuu^2}{3^3\cdot 2^2}\Big)=
{2}\sum_{k=0}^\infty\frac{(4k+1)!}{(3k+2)!k!}\Big(-\frac{\uuu^2}{4^5}\Big)^k
\en
has the closed form expression ${\mathfrak F}(\uuu^2)=\frac{32}{\uuu}t(\uuu)$ where $t(\uuu)$
is given in (\ref{radical1}),(\ref{radical2}), and, because of
(\ref{poleq4}), that it satisfies the ``hidden'' identity
\eq
{\mathfrak F}(\uuu^2)=\Big(1-\frac{\uuu^2}{4^5}{{\mathfrak F}(\uuu^2)}^2\Big)^2~.\label{qeq}
\en
 It is indeed this
identity that we have verified
up to $O(\uuu^{1000})$.
\subsection{General nonlinear theory}
Since Born-Infeld theory is singled out by setting
$r=\frac{\la}{8}|T\HT|$, and Maxwell theory by setting $r=0$
(cf. Example \ref{EEEEX}), it is
convenient to describe a general nonlinear theory without higher
derivatives by setting
\eq\label{fuuu}
r=\frac{\la}{8}|T\HT|f(\uuu)/\uuu
\en
 where $f(\uuu)$ is a positive
function of $\uuu$. We require the theory to reduce to
electromagnetism in the weak field limit, i.e.,
$\HG_{\mu\nu} =-F+ o(F)$ for $F\to 0$. Then we have $T^-={\cal O}(F)$,
$T^+=o(F)$, $u={\cal O}(F^2)$. Hence from (\ref{tp}) we obtain $\lim_{u\to 0} t=0$. Moreover from
(\ref{trrel}),  $r={\cal O}(t)$ and from $r=\frac{1}{16}f(u)(1-t^2)$
(that follows from (\ref{fuuu}) and (\ref{useful})\/) $f={\cal
  O}(t)$.
Hence the theory reduces to
electromagnetism in the weak field limit if and only if $\lim_{\uuu\to
  0}f(\uuu)=0$.\footnote{We further notice that  $\lim_{\uuu\to
  0}f(\uuu)=0$ implies $I(u)=o(u)$. In particular the theory defined
by $I(u)=u$ (or equivalently $f(u)=\frac{2^6}{3^2}$) does not reduce to
electromagnetism for weak fields.

In general, besides requiring that
the theory determined by $f(u)$ reduces to electromagnetism in the
weak field limit we can also require the theory to be analytic in $F$
(the lagrangian to have a power series expansion in $F$ around $F=0$). In
this case from (\ref{LegendreT}) and inverting relations (\ref{TUF}), or
more explicitly from (\ref{Scorrec}), we see that  the Legendre transformed
function $I(u)$ must depend on $u^2=4\la^2 { T^-}^2 {\overline{T^-}}^2$. Equivalently $f(u)/u$ must depend
on $u^2$.}

From  $r=\frac{1}{16}f(u)(1-t^2)$ (that follows from (\ref{fuuu}) and (\ref{useful})\/)  and (\ref{trrel}) we obtain that the composite function $t(f(\uuu))$ satisfies the fourth order polynomial equation 
\eq
(1-t^2)^2f(\uuu)=32 t~, \label{TFU}
\en
so that
$t(f(\uuu))$ is obtained with the  substitution $\uuu\to f(\uuu)$ in
(\ref{radical1}) and (\ref{radical2}), or in
(\ref{3F2}).

More explicitly, recalling the constitutive relation (\ref{1/TT}), we
conclude that the constitutive relations \`a la   Schr\"odinger 
\eq\label{cr5}
\HT{}_{\!\mu\nu}=-\frac{T^2}{T \HT} T_{\mu\nu}-\frac{\la}{8}\frac{f(\uuu)}{\uuu}_{\,}(T\HT)\,\bar T_{\mu\nu}~,
\en
are equivalent to the constitutive relations (deformed twisted
self-duality conditions)
\eq\label{cr6}
{T^+}^{\mu\nu}
=
\frac{1}{2\la} t(f(\uuu)) \,\frac{\partial \uuu}{\partial{\overline{T^-}_{\!\!\mu\nu}}}~,
\en
where $t(f(u))$ satisfies the quartic equation (\ref{TFU}), and we
recall that $\uuu = 2\la|{T^-}^{2}| = \la(|T^2|+|T\HT|)~.$ 
\sk
In other words the appearence of the quartic equation (\ref{TFU})  is a
general feature of the relation between the constitutive relations (\ref{cr5})
and (\ref{cr6}), it appears for any self-dual theory and it is not only a feature
of the Born-Infeld theory.

\sk\sk
\noi
{\large \bf Acknowledgements} \\
We thank W. Chemissany, R. Kallosh, T. Ortin and  M. Trigiante for
valuable correspondence during last fall.  In particular we thank R. Kallosh for discussions and her interest on the relation
between the hypergeometric function of her work \cite{CKR} and the quartic 
equation 
expressing Born Infeld theory in duality invariant variables.
We also thank F. Bonechi for fruitful discussions.

The hospitality of CERN Theory Unit where the present work has been initiated is  gratefully acknowledged.
This work is supported by the ERC Advanced Grant no. 226455, Supersymmetry, Quantum Gravity and Gauge Fields (SUPERFIELDS).
\begin{appendix}\section{Examples of higher derivatives theories\label{HDactions}}
We construct  examples of higher derivatives $U(1)$ actions that define
self-dual theories.
These examples include the Bossard Nicolai one
in \cite{BN}; the actions we present are quadratic in the field strength.
Let
\eq
S=-\frac{1}{4}\int\!\! d^4x\sqrt{g} \; F O F\label{FOF}
\en
with $O$ a matrix 
$O_{ \mu\nu}^{~\rho\sigma}$   
of differential operators independent from $F$; explicitly
$F O F=
F^{\mu\nu}\,O_{ \mu\nu}^{~\rho\sigma} F_{\rho\sigma}\,.$ 
We recall that by definition the hermitian conjugate
operator $O^\dagger$ satisfies 
$\int (O^\dagger K)F=\int KOF$ for all antisymmetric and real tensors
$K$ and $F$.
Since $\int F O F =\int (O^\dagger F)  F = \int F (O^\dagger F)
$,  there is no restriction in
considering  $O$ hermitian,
i.e., $\int (OK) F=\int KOF$, or explicitly
$\int \!d^4x \sqrt{g} \,(O_{\rho\sigma}^{~\mu\nu}
K_{\mu\nu})F^{\rho\sigma}=
\int\! d^4x \sqrt{g}\, K^{\mu\nu}O^{~\rho\sigma}_{\mu\nu} F_{\rho\sigma}\,.$
Let $O$ also satisfy
\eq
 O\circ {}^{\ast}\circ O={}^{\ast}\label{OtOt}
\en
i.e., $O\,{}^{\ast}(OF)={}^{\ast}F$.
\sk
We show
that the action (\ref{FOF}) gives self-dual equations of motion
if $O$ satisfies (\ref{OtOt}).
Indeed in this case  the self-duality condition
(\ref{NGZ2}), i.e., 
$\int \!d^4 x \;F\widetilde F +G\widetilde G=0$, holds.  The proof is
easy. We first calculate
\eq\widetilde G^{\mu\nu}= 2
\frac{\delta S}{\delta F_{\mu\nu}}=-\sqrt{g}\,O^{\mu\nu\;\rho\sigma}
F_{\rho\sigma}~,\label{GOFcr}
\en
i.e., $\HG{}^{\,\mu\nu}= -(OF)^{\mu\nu}$. Hence
\eqa
\int \!d^4x\,\widetilde G G&=& \int \!d^4x\sqrt{g}\;(\HG)\,G=
-\!\int
\!d^4x\sqrt{g}\;{(\HG)}^{\:\ast}({\HG})\nn\\
&=&
-\!\int \!d^4x\sqrt{g}\, (OF)~{}^{\ast\!}(OF)=
-\!\int \!d^4x\sqrt{g}\, F O\;{}^{\ast}\!(OF)\nn\\
&=&
-\!\int \!d^4x\sqrt{g}\, F \HF=
-\!\int \!d^4x \,F \widetilde F~,
\ena
where in the fourth equality we used (\ref{OtOt}).

\sk
Examples of differential operators $O$ are given by considering  operators
$\Delta$ on antisymmetric tensors $F_{\mu\nu}$  that satisfy the
hermiticity condition $\Delta^\dagger=\Delta$ and that anticommute
with the $\ast$-Hodge operator,
\eq
{}^{\ast}\circ \Delta=-\Delta\circ {}^{\ast}~.
\en
Let's introduce a coupling constant $\la$ so that  $\la \Delta$ is adimensional, and let $f(\la \Delta)$ be an odd function in $\Delta$ (e.g.,  a polynomial,
or  a power series function like $\la\Delta$, $\la\Delta^3$,
$\sin(\la\Delta)$). Then ${}^{\ast}\circ
f(\la\Delta)=-f(\la\Delta)\circ {}^{\ast}$,
and 
\eq
O=
\big(1-f(\la\Delta)\big)^{-1}\big(1+f(\la\Delta)\big)
\en
satisfies (\ref{OtOt}).

In particular, if $f(\la\Delta)=\la\Delta$ and if
$\Delta_{\mu\nu}^{~\rho\sigma} F_{\rho\sigma}=\nabla_{\kappa}
\big(T_{[\mu}^{~~\kappa\la[\sigma}\nabla_\la\delta_{\nu]}^{\rho]}F_{\rho\sigma}\big)$,
where the covariant derivatives are with respect to the Levi-Civita
connection,  $T^{\mu\kappa\la\sigma}$ is the
Bel-Robinson tensor, and the square brackets denote antisymmetrization
in the embraced indices, 
then we obtain the action of Bossard and Nicolai \cite{BN}.

\section{The action functional $S[F]$ from 
${\cal I}[T^-,\overline{T^-}]$\label{app1}}
We here determine the first two nontrivial terms  $S^{(1)}$ 
and  $S^{(2)}$  of the action $S$, see (\ref{Scorrec}) Section
\ref{ITTm}.

Since $S^{(0)}=-\frac{1}{4}\int \!d^4x\sqrt{g}\,F^2$ corresponds to 
${\cal I}^{[0]}=0$, we have (cf.  (\ref{al1})) , 
${T^+}^{(0)}=0\,,~{T^-}^{(0)}=2F^-\,,~{G^-}^{(0)}=iF^-$, and, for
$n\geq 1$, ${T^-}^{(n)}=-i{G^-}^{(n)}$,
$\overline{T^-}^{(n)}=i{G^+}^{(n)}$.
The following useful formula is then easily derived using the chain rule:
\eq\label{Imn}
\frac{\delta {\cal I}^{[m]}|_{F^\mp}^{~(n)}}{\delta
F^-}
=2\frac{\delta {\cal I}^{[m]}}{\delta
T^-}\Big|_{F^\mp}^{\;(n)}-2\sum_{p=m}^{n-1}\int\!\!d^4x\frac{1}{\sqrt{g}}\,
\frac{\delta {\cal I}^{[m]}}{\delta
T^-}\Big|_{F^\mp}^{\;(p)}\,\frac{\delta^2S^{(n-p)}}{{\delta
 F^-}{\delta F^-}}+
\frac{\delta {\cal I}^{[m]}}{\delta
\overline{T^-}}\Big|_{F^\mp}^{\;(p)}\,\frac{\delta^2S^{(n-p)}}{{\delta
 F^-}{\delta F^+}}
\en
where we have simplified the notation by setting
$\big|_{F^\mp}=\big|_{{{{}^{T^-[F^-,F^+]}_{\overline{T^-}[F^-,F^+]}}}}\,$
and omitting spacetime indices,
and where we have assumed that we know the action $S[F]$ up to order
$n-1$ so that, for all $p=1,2,\ldots n-1$, we have
$\mp {iG^\pm}^{(p)}=\frac{2}{\sqrt{g}}
\frac{\delta^2 S^{(p)}}{\del F^\pm}$,
and therefore 
$\frac{\delta {T^-}^{(p)}}{\delta
  F^-}=-i\frac{\delta {G^-}^{(p)}}{\delta F^-}=-\frac{2}{\sqrt{g}}
\frac{\delta^2 S^{(p)}}{\del F^-\del F^-}$.
\sk
If $m=n$ then the above formula simply reads 
$\frac{\delta {\cal I}^{[n]}|_{F^\mp}^{~(n)}}{\delta
F^-}
=2\frac{\delta {\cal I}^{[n]}}{\delta
T^-}\Big|_{F^\mp}^{\;(n)}$, and since 
${\cal I}^{[n]}|_{F^\mp}^{~(n)}={\cal I}^{[n]}[2F^-,2F^+]$ (use ${T^-}^{(0)}=2F^-$),
it 
simplifies to 
\eq\label{Inn}
\frac{\delta {\cal I}^{[n]}[2F^-,2F^+]}{\delta
F^-}
=2\frac{\delta {\cal I}^{[n]}}{\delta
T^-}\Big|_{F^\mp}^{\;(n)}~.
\en
Setting $n=1$ and recalling that since ${\cal I}^{[0]}=0$,
$\frac{\delta {\cal I}^{[1]}}{\delta
T^-}|_{F^\mp}^{\;(1)}=\frac{\delta {\cal I}\,}{\delta
T^-}|_{F^\mp}^{\;(1)}\,$, we immediately see
that  $S^{(1)}[F^-,F^+]=\frac{1}{4}{\cal I}^{[1]}[2F^-,2F^+]$
satisfies (\ref{recS}). 
\sk
In order to determine $S^{(2)}$ we first calculate
(using for example the chain rule in deriving w.r.t. $\la$)
\eqa\label{I12}
{\cal I}^{[1]}\big|_{F^\mp}^{\;(2)}&=&\int \!\!d^4x~\frac{\delta {\cal I}^{[1]}}{\delta
T^-}\Big|_{F^\mp}^{\;(1)}{T^-}^{(1)}
+\frac{\delta {\cal I}^{[1]}}{\delta\overline{T^-}}
\Big|_{F^\mp}^{\;(1)}
\overline{T^-}^{(1)}\nn\\[.4em]
&=&2\int \!\!d^4x~
\frac{\delta {S}^{(1)}}{\delta
F^-}(-iG^-)^{(1)}+
\frac{\delta {S}^{(1)}}{\delta {F^+}} (iG^+)^{(1)}\nn\\[.4em]
&=&-4\int \!\! d^4x \frac{1}{\sqrt{g}}~
\frac{\delta {S}^{(1)}}{\delta
F^-}\frac{\delta {S}^{(1)}}{\delta
F^-}+
\frac{\delta {S}^{(1)}}{\delta {F^+}}
\frac{\delta {S}^{(1)}}{\delta {F^+}}
\ena
where  in the second line we used $\frac{\delta {\cal I}^{[1]}}{\delta
T^-}|_{F^\mp}^{\;(1)}=\frac{\delta {\cal I}\,}{\delta
T^-}|_{F^\mp}^{\;(1)}\,$ and then
(\ref{recS}) at order $n=1$.
In the third line we used 
the constitutive relations (\ref{Sconst}), i.e., $G^-=-\frac{2i}{\sqrt{g}}\frac{\delta
  S}{\,\delta F^-}$ at order $n=1$, that we already know to be
implied by the chiral constitutive relations (\ref{Iconst}).

\sk
Next for notational simplicity we set
$\int=\int\!d^4x\frac{1}{\sqrt{g}}\,$ and we compute
\eqa
\frac{\delta {\cal I}^{}}{\,\delta
T^-}\Big|_{F^\mp}^{\;(2)}
&=&
\frac{\delta {\cal I}^{[2]}}{\delta
T^-}\Big|_{F^\mp}^{\;(2)}
+\frac{\delta {\cal I}^{[1]}}{\delta
T^-}\Big|_{F^\mp}^{\;(2)}\nn\\[.6em]
&=&
\frac{1}{2}\frac{{\cal I}^{[2]}[2F^-,2F^+]}{\delta F^-}
+\frac{1}{2}\frac{\delta{\cal I}^{[1]}|_{F^\mp}^{\;(2)}}{\delta F^-}
+\int
\frac{\delta {\cal I}^{[1]}}{\delta
T^-}\Big|_{F^\mp}^{\;(1)}\frac{\delta^2S^{(1)}}{\delta F^-\delta F^-}
+
\frac{\delta {\cal I}^{[1]}}{\delta
\overline{T^-}}\Big|_{F^\mp}^{\;(1)}\frac{\delta^2S^{(1)}}{\delta F^-\delta F^+}\nn\\[.6em]
&=&
\frac{1}{2}\frac{{\cal I}^{[2]}[2F^-,2F^+]}{\delta F^-}
+\frac{1}{2}\frac{\delta{\cal I}^{[1]}|_{F^\mp}^{\;(2)}}{\delta F^-}
+\frac{\delta}{\delta F^-}\int
\frac{\delta S^{(1)}}{\delta F^-}\frac{\delta S^{(1)}}{\delta F^-}+
\frac{\delta S^{(1)}}{\delta F^+}\frac{\delta S^{(1)}}{\delta
  F^+} \nn\\[.6em]
&=&\frac{\delta}{\delta F^-}\Big(\frac{1}{2}{\cal I}^{[2]}[2F^-,2F^+]-\int
\frac{\delta S^{(1)}}{\delta F^-}\frac{\delta S^{(1)}}{\delta F^-}+
\frac{\delta S^{(1)}}{\delta F^+}\frac{\delta S^{(1)}}{\delta
  F^+}\,\Big)\label{IdF}
\ena
where in the second line we have used (\ref{Inn}) and (\ref{Imn}), in
the third line we have noticed again that 
$\frac{\delta {\cal I}^{[1]}}{\delta
T^-}|_{F^\mp}^{\;(1)}=\frac{\delta {\cal I}\,}{\delta
T^-}|_{F^\mp}^{\;(1)}=2\frac{\delta S^{(1)}}{\delta F^-}$ (cf. (\ref{recS}), in the
fourth line we have used (\ref{I12}). From the equality (\ref{IdF}) we
see that  
$S^{(2)}=\frac{1}{4}_{\!\,}{\cal I}^{[2]}[2F^-,2F^+]-
\frac{1}{2}\!\int\frac{\delta S^{(1)}}{\delta F^-}\frac{\delta S^{(1)}}{\delta F^-}+
\frac{\delta S^{(1)}}{\delta F^+}\frac{\delta S^{(1)}}{\delta
  F^+}$ satisfies (\ref{recS}) with $n=2$.

\section{The energy momentum tensor and its trace \label{TTT}}
We first recall that the symmetric energy-momentum tensor $\theta^{\mu\nu}$
of a nonlinear electromagnetic theory
is  given by
\eq
\theta^{\mu\nu}=-\widetilde G^{\mu\la}F^{\nu}_{~~\la}+g^{\mu\nu\,}{\cal L} \label{emt}
\en
if the Lagrangian $L$ in the action  $S[F]=\int \!d^4x \,{\cal L}=\frac{1}{\la}\int \!d^4x \sqrt{g}\,
L$
depends on the field strength $F_{\mu\nu}$ and the metric $g_{\mu\nu}$
only via the invariant and dimensionless combinations
\eq
\al=\frac{\la}{4}F^2~~,~~~\be=\frac{\la}{4}F\HF\label{albe}~.
\en
Indeed we  compute 
\eq
\frac{\partial\al}{\partial g_{\mu\nu}}=-2\frac{\partial \al}{\partial
  F_{\mu\rho}}F^\nu_{~\;\rho}~~,~~~
\frac{\partial\be}{\partial g_{\mu\nu}}=-2\frac{\partial \be}{\partial F_{\mu\rho}}F^\nu_{~\;\rho}~,~~
\en
(where the factor $2$ is due to our $\frac{\partial}{\partial
  F_{\mu\rho}}$ conventions, cf. (\ref{Sconst}) and its footnote); for
the second equation we used
$\frac{\partial\sqrt{g}^{-1}}{\partial g_{\mu\nu}}=-\sqrt{g}^{-1}g^{\mu\nu}$,
and the property $\HF^{\mu\la}
F_{\nu\la}=-\frac{1}{4}\delta^{\mu}_{~\nu\,}\HF^{\rho\sigma}
F_{\rho\sigma}$. Expression (\ref{emt}) for the energy momentum tensor
$\theta^{\mu\nu}=\frac{\delta S}{\delta g_{\mu\nu}}$ 
is then straightforward. 

Now an action in Minkowski spacetime that has no derivatives of the
field strength $F$, by Lorentz invariance depends on $F$ only via the
(pseudo)scalars $F^2$ and $F\widetilde F$. We can then always
minimally couple  the action to gravity  so that the metric enters
only in (\ref{albe}), and hence so that (\ref{emt}) holds. Even if the
coupling to gravity (for example in order to preserve symmetry properties)
requires terms like $RF^2$ where $R$ is the scalar curvature,
expression (\ref{emt}) still holds in flat spacetime.

\sk
From (\ref{emt}) it follows that the trace of the energy momentum
tensor satisfies
\eq
\frac{1}{4}\theta^\mu_{~~\mu}={\cal L}-\frac{1}{4}\widetilde G F~.
\en
We therefore  have
\eq\label{D5}
\frac{1}{4}\int \!d^4x \,g_{\mu\nu}\frac{\delta S}{\delta
  g_{\mu\nu}}=\int\!d^4x \,\frac{1}{4}\theta^\mu_{~~\mu}=S-\frac{1}{4}\int\!d^4x\,\widetilde G F=
-\la\frac{\partial S}{\partial \la}~,
\en
the last relation follows
observing that the inverse metric $g^{\mu\nu}$ appears always with the
factor $\la^{1/2}$ in the action  $S[F]=\int \!d^4x \,{\cal L}=\frac{1}{\la}\int
\!d^4x \sqrt{g}\, L$ (cf. (\ref{albe})).

Finally if we let $S[F]\to S_c[F]=\frac{1}{c^2}S[cF]$, we see that  (\ref{D5}) coincides with 
 (\ref{Ttrace}). Indeed  $\la\frac{\partial}{\partial \la}$ equals
 $c^2\frac{\partial}{\partial c^2}$ because $S_c[F]$ 
 depends only on the product $c^2\la$.

\end{appendix}


\end{document}